\documentclass[journal]{IEEEtran}
\usepackage{amssymb}
\usepackage{amsmath}
\usepackage{graphicx}
\usepackage{cite}
\usepackage{url}

\usepackage{color}
\usepackage{citesort}
\usepackage{multicol}
\usepackage{amsfonts}
\usepackage{bbm}
\usepackage{algorithmic}
\usepackage{bm}
\usepackage{subfigure}
\usepackage{graphicx,epstopdf}
\usepackage{epsfig}	
\usepackage{cite,graphicx,amsmath,amssymb}
\usepackage{comment}
\usepackage{amsmath}
\usepackage{mathabx}
\usepackage{lipsum}
\usepackage{stfloats}
\usepackage{algorithm}
\usepackage{algorithmic}

\newlength{\halfpagewidth}
\divide\halfpagewidth by 2

\newtheorem{theorem}{\textbf{Theorem}}

\newtheorem{lemma}{\textbf{Lemma}}

\newtheorem{corollary}{\textbf{Corollary}}

\newtheorem{proof}{\textbf{Proof}}

\newtheorem{definition}{\textbf{Definition}}

\newtheorem{proposition}{\textit{Proposition}}

\makeatletter
\def\ScaleIfNeeded{%
\ifdim\Gin@nat@width>\linewidth \linewidth \else \Gin@nat@width
\fi } \makeatother

\begin{document}
%\pagestyle{fancyplain}
%
%\pagestyle{fancy}
%\lhead[]%
 %   {\footnotesize Physical layer security}
%\cfoot{}

\title{\Huge{Content Placement in Cache-Enabled {\color{black}Sub-6 GHz} and Millimeter-Wave Multi-antenna Dense Small Cell Networks}}

 \author{ Yongxu Zhu, Gan Zheng,~\IEEEmembership{Senior Member, IEEE}, Lifeng Wang,~\IEEEmembership{Member, IEEE}, Kai-Kit Wong,~\IEEEmembership{Fellow, IEEE}, and Liqiang Zhao,~\IEEEmembership{Member, IEEE}
 %\thanks{Manuscript received June 15, 2017; revised October 03, 2017; accepted January 9, 2018. This work was supported by the U.K. Engineering and Physical Sciences Research Council (EPSRC) under Grant EP/N007840/1 and also supported by the National Natural Science Foundation of China (61771358). A part of this paper was presented at IEEE Global Telecommunications Conf. (GLOBECOM), Singapore, December 2017 \cite{yongxu_GC}. The editor coordinating the review of this manuscript and approving it for publication was Dr. S. Mukherjee. (\textit{Corresponding author: Kai-Kit Wong})}
 \thanks{Y. Zhu and G. Zheng are  with the Wolfson School of Mechanical, Electrical and Manufacturing Engineering, Loughborough University, Leicestershire, LE11 3TU, UK (Email: \{y.zhu4, g.zheng\}@lboro.ac.uk).}\thanks{L. Wang and K.-K. Wong are with the Department of Electronic and Electrical Engineering, University College London, London, WC1E 6BT, UK (Email: \{lifeng.wang, kai-kit.wong\}@ucl.ac.uk).}
\thanks{L. Zhao is with State Key Laboratory of Integrated Service Networks, Xidian University, Xi’an 710071, China and also with the National Mobile Communications Research Laboratory, Southeast University, Nanjing 210096, China (Email: lqzhao@mail.xidian.edu.cn).}
}
 
\maketitle
\vspace{-15mm}
\begin{abstract}
This paper studies the performance of cache-enabled dense small cell networks consisting of multi-antenna \textcolor{black}{sub-6 GHz} and millimeter-wave base stations. \textcolor{black}{Different from the existing works which only consider a single antenna at each base station, the optimal content placement is unknown when the base stations have multiple antennas.} We first derive the successful content delivery {probability} {by accounting for the key channel features  at sub-6 GHz and mmWave frequencies}.    \textcolor{black}{
The maximization of the successful content delivery probability is a challenging problem. To tackle it, we first propose a constrained cross-entropy algorithm which achieves the  near-optimal solution with moderate complexity. We then develop another  simple yet effective heuristic probabilistic content placement scheme, termed two-stair algorithm, which strikes a balance between caching the most popular contents and achieving content diversity.}  \textcolor{black}{Numerical results demonstrate the superior performance of the  constrained  cross-entropy method and that the two-stair algorithm yields significantly better performance than {only} caching the most popular contents.} The comparisons between the \textcolor{black}{sub-6 GHz} and {mmWave} systems  reveal  an interesting tradeoff between caching capacity and density for the {mmWave} system to achieve  similar performance as the \textcolor{black}{sub-6 GHz} system.
\end{abstract}

%It is  remarkable to see that increasing caching  capacity can effectively reduce the density requirement on the millimeter-wave system in order to achieve similar performance as the microwave system.

%

%========================================================================
\section{Introduction}
The global mobile data traffic continues growing at an unprecedented pace and will reach 49 exabytes monthly by 2021,  of which 78 percent will be video contents\cite{Cisco}.
To meet the high capacity requirement for the future mobile networks, one promising solution is network densification, i.e., deploying dense small cell base stations (SBSs) in the existing \textcolor{black}{macrocell} cellular networks. \textcolor{black}{Although large numbers of small cells shorten the communication distance, the major challenge is to transfer the huge amount of mobile data from the core networks to the small cells and this imposes stringent demands on backhaul links.} To address this problem, caching popular contents at small cells has been proposed as one of the most effective solutions, considering the fact that most mobile data are  contents such as video, weather forecasts, news and maps, that are repeatedly requested and cacheable \cite{fang2014survey}. The combination of small cells and caching will bring content closer to users, decrease backhaul traffic  and reduce transmission delays, thus alleviating many bottleneck problems in wireless content {delivery} networks. {\color{black}This paper focuses on the caching design at
both sub-6 GHz ($\mu$Wave) and millimeter-wave ($\mathrm{mm}$Wave)\footnote{{\color{black}In this paper, we focus on $\mathrm{mm}$Wave frequencies from 30 GHz to 300 GHz.}} SBSs in dense small cell networks.}

\subsection{Related Works}
\subsubsection{Caching in $\mu$Wave and $\mathrm{mm}$Wave networks}
{\color{black}{ MmWave communication has received much interest for providing high capacity because there are vast amount of inexpensive spectra available in the 30 GHz-300 GHz range. However, compared to $\mu$Wave frequencies, mmWave channel experiences excessive attenuation due to rainfall, atmospheric or gaseous absorption, and is susceptible to blockage. To redeem these drawbacks, mmWave small cells need to adopt narrow beamforming and be densely deployed in an attempt to provide seemless coverage~\cite{Larsson_mimo,Heath_comparison,wang_magazine}. The study of content caching applications in mmWave networks is of great importance, due to the fact that mmWave will be a key component of future wireless access and content caching at the edge of networks is one of 5G service requirements~\cite{wang_magazine}. Cache assignment with video streaming in mmWave SBSs on the highway is discussed in \cite{qiao2016proactive} and it is shown to significantly reduce the connection and retrieval
delays.}} Certainly, combining the advantages of
\textcolor{black}{$\mu$Wave} and $\mathrm{mm}$Wave technologies will bring more benefits \cite{rois2016heterogeneous}. Caching in dual-mode SBSs that integrate both \textcolor{black}{$\mu$Wave} and $\mathrm{mm}$Wave  frequencies is studied in \cite{semiari2017caching}, where dynamic matching game-theoretic approach is applied to maximize the handovers to SBSs in the mobility management scenarios. %using a dynamic matching game.
The proposed methods can minimize handover failures and reduce  energy consumption  in highly mobile heterogeneous networks. {\color{black}{Dynamic traffic in cache-enabled network was studied in \cite{xia2016modeling}.}}

{\color{black}{Recent contributions also pay attention to the caching in MIMO networks such as \cite{liu2015exploiting,liu2014cache,yang2016content}. In \cite{liu2015exploiting,liu2014cache}, cache-enabled cooperative MIMO framework for wireless video streaming is investigated. In \cite{yang2016content}, coded caching for downlink MIMO channel is discussed.}}

\subsubsection{Optimization of content placement}
 Content placement
 %given storage constraints
 {with finite cache size} is the key issue in caching design, since unplanned caching in nearby SBSs will %only generate
 {result in} more interference. The traditional method of caching most popular content (MPC) in wired networks is no longer optimal when considering the wireless transmission. A strategy that combines MPC and the largest content diversity  caching is proposed in \cite{chen2016cooperative}, together with cooperative transmission in cluster-centric small cell networks. This strategy is   extended to the distributed relay networks  with relay clustering in \cite{zheng2016optimization} to  combat the half-duplex constraint, and it significantly improves the  outage performance.
% Caching distribution in device-to-device (D2D) communication networks is studied in \cite{malak2016optimizing}  under a simple transmission strategy where a single file is transmitted at a time. \textcolor{black}{In that paper,} the optimal caching distribution to maximize the density of successful receptions is found to follow Zipf distribution when the file request is also a Zipf distribution. Several spatially correlated caching strategies are investigated for D2D networks in   \cite{malak2016spatially} in contrast to the baseline independent content placement.  It is shown that Mat\'{e}rn hard-core model based  content placement often yields a higher hit probability and more suitable for the D2D setting.
 A multi-threshold caching that allows {BSs} to store different number of copies of contents according to their popularity is proposed in  \cite{ao2015distributed}, and it allows a {finer partitioning} of the cache space than binary threshold, but its complexity is exponential in the number of thresholds.

 {\color{black}{Probabilistic content placement under random network topologies has also been investigated.}} In \cite{Blaszczyszyn_2015},  the optimal content caching probability that maximizes the hit probability %is found in closed-form
 \textcolor{black}{is derived}. The results are
 extended to heterogeneous cellular networks in \cite{HCN_prob}  which shows that caching the most popular contents in the macro BSs is almost optimal while it is in general not optimal for SBSs.

\subsubsection{Caching in Heterogeneous Networks}
Extensive works have been carried out to understand the performance gain of  caching for  heterogeneous   networks (HetNets) and stochastic geometry is the commonly used approach.
In \cite{li2016optimization}, the optimal probabilistic caching to maximize the successful delivery probability is considered in
% an N-tier wireless heterogeneous network  in the high signal-to-noise ratio (SNR) region
{a multi-tier HetNet}. The cache-enabled heterogeneous signal-antenna cellular networks are investigated in \cite{wen2016cache}. The optimal probabilistic content placement for the interference-limited cases is derived, and the result shows that the optimal placement probability is linearly proportional to the square root of the content popularity with an offset depending on BS caching capabilities.  %The idea of joint SBS caching and BSs cooperation is proposed in \cite{wen2017random} for the downlink of a large-scale HetNet to maximize the successful transmission probability.
 Caching policies to maximization of  success probability and area spectral efficiency of cache-enabled HetNets are studied in \cite{liu2016cache}, and the results show that the optimal caching probability is less  skewed to maximize the success probability but is more skewed to maximize the area spectral efficiency.
{The work of \cite{wen2017random} proposes a joint BS caching and cooperation for maximizing the successful transmission probability in a multi-tier HetNet.} A local optimum  is obtained in the general case and  global optimal solutions are achieved in some special cases.
 Cache-based channel selection diversity and network interference are studied in \cite{chae2016caching} in stochastic wireless caching helper networks, and solutions for  noise-limited networks and  interference-limited networks are derived, respectively.

\subsection{Contributions and Organization}
 The existing caching design for SBSs are restricted to the single-antenna case and mainly for the $\mu$Wave band. Little is known about the impact of multiple antennas at the densely deployed SBSs and the adoption of $\mathrm{mm}$Wave band on the successful content delivery and the optimal content placement. \textcolor{black}{Analyzing multi-antenna networks using stochastic geometry is a known difficulty, as acknowledged in \cite{yu2017tractable}}. In contrast to existing works, in this paper we analyze the performance of caching in multi-antenna SBSs in $\mu$Wave and $\mathrm{mm}$Wave networks, and propose  probabilistic content placement schemes to maximize the performance of content delivery. The main contributions of this paper are summarised as follows:
\begin{itemize}
\item
 {Derivation of successful content delivery probability (SCDP) of multi-antenna SBSs.} We use stochastic geometry to model wireless caching in multi-antenna dense small cell networks in both $\mu$wave and $\mathrm{mm}$wave bands. The SCDPs for both types of cache-enabled SBSs are derived. The results characterize the dependence of  the  SCDPs on parameters such as channel effects, caching placement probability, SBS density, transmission power and number of antennas.
\item
\textcolor{black}{Development of a near-optimal  cross-entropy optimization (CEO) method for a general distribution of content requests.  The derived SCDPs do not admit a closed form, and are highly complex to optimize. To tackle this difficulty, we first propose  a constrained CEO (CCEO) based algorithm that optimizes the SCDPs. The original unconstrained CEO algorithm is a stochastic optimization method based on  adaptive importance sampling that can achieve the near-optimal solution with moderate complexity and      guaranteed convergence  \cite{CEO}. We adapt this method to  deal with the caching capacity constraints and the probabilities constraints in our problem. }
\item
 {Design of a simple heuristic  content placement algorithm.} To further reduce the complexity, we propose a heuristic two-stage algorithm  to maximize the SCDP via probabilistic content placement when the content request probability follows the {Zipf} distribution \cite{Zipf}. The algorithm is designed by combining MPC and caching diversity (CD) schemes while taking into account the content popularity. The solution  demonstrates near-optimal performance in single-antenna systems, and various  advantages in  multi-antenna scenarios.
\item Numerical results  show that in contrast to the traditional way of deploying much higher density  SBSs or installing many more antennas, increasing caching capacity at $\mathrm{mm}$Wave SBSs provides a low-cost solution to  achieve comparable SCDP performance as $\mu$Wave systems.
\end{itemize}

The rest of this paper is organized as follows. The system model is presented in Section II. The analysis of SCDPs for $\mu$Wave and $\mathrm{mm}$Wave systems are provided in Section III. Two probabilistic content placement schemes are described in Section IV. Simulation and numerical results as well as discussions are given in Section V, followed by concluding remarks in Section VI.

\section{System Model}
We consider a cache-enabled dense small cell networks consisting of the $\mu$Wave and $\mathrm{mm}$Wave SBSs tiers. In such networks, {\color{black}{each user equipment (UE) in a tier is associated with the nearest SBS that has cached the desired content, and the optimal designs of content placement under such association assumption can address the concern that operators are required to place the content caches close to UEs~\cite{3GPP}.}} We assume that there is a finite content library denoted as $\mathcal{F}:=\lbrace f_1,\dots,f_j,\dots,f_J \rbrace$, where $f_j$ is the $j$-th most popular content and the number of {contents} is $J$, \textcolor{black}{we assume each content has normalized size of 1 and each BS can only store up to $M$ contents \cite{ao2015distributed,wen2016cache,chae2016caching}. The {analysis and optimization }can be applied to the case of unequal content sizes.} It is assumed that $M \ll J$. {\color{black}{The request probability for the $j$-th content is $a_j$, and $\sum\nolimits_{j{\rm{ = }}1}^J {{a_j}} {\rm{ = }}1$.}} Without loss of generality, we assume the contents are sorted according to a descending order of  $a_j$.

\subsection{Probabilistic Content Placement}
\begin{figure}[t!]
\centering
\includegraphics[width=1.5 in]{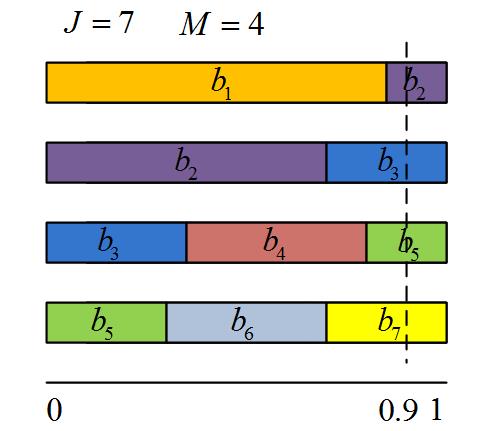}
\caption{\label{fig:system}\textcolor{black}{Probabilistic content placement strategy.}}
\label{sys}
\end{figure}
We consider a  probabilistic caching model where the content is independently stored with the same probability in all SBSs of the same tier (either $\mu$Wave or $\mathrm{mm}$Wave) \cite{Blaszczyszyn_2015}. Let $b_j$ denote the probability that the $j$-th content is cached at a SBS. Fig. \ref{fig:system} shows an example of probabilistic caching with $J=7$ and $M=4$, where the contents $\left\{ f_2, f_3, f_5, f_7 \right\}$ are cached at a SBS by drawing uniformly a random
number which is 0.9 in this example.
In the probabilistic caching strategy, the caching probability \textcolor{black}{${\bm{b}}=\{b_1,...b_j,...b_J\}$} needs to satisfy the following conditions:
\begin{align}\label{Prob_caching}
& \sum\limits_{j{\rm{ = }}1}^J {{b_j}} {\rm{ \leq }}M, \nonumber\\
& 0 \leq b_j \leq 1, \; \forall j.
\end{align}

\textcolor{black}{Note that although the probabilistic caching strategy is used, implementation of it will allow each SBS to always cache the maximum amount of total contents up to its caching capacity $M$.}

\subsection{Downlink Transmission}
In the considered downlink networks, each $\mu$Wave SBS is equipped with $N_\mu$ antennas, and each $\mathrm{mm}$Wave SBS has directional $\mathrm{mm}$Wave antennas. All UEs are single-antenna nodes, \textcolor{black}{in the both $\mu$Wave and $\rm{mm}$Wave, only one single-antenna user is allowed to communicate with the SBS at one time slot\footnote{\textcolor{black}{In dense small cell networks, we assume that the density of users is much higher than the density of $\mu$Wave or $\rm{mm}$Wave SBSs and this can be handled by using multiple access techniques \cite{zhu2016wireless}.}}.} The positions of $\mu$Wave SBSs are modeled by  a homogeneous Poisson point process (HPPP) $\Phi^\mu$ with the density $\lambda_\mu$, and the positions of $\mathrm{mm}$Wave SBSs are modeled by   an independent HPPP $\Phi^\mathrm{mm}$ with the density $\lambda_\mathrm{mm}$. Define ${\Phi^\mu_j}$ and ${\Phi^{\rm{mm}}_j}$ as the point process corresponding to all SBSs that cache the content $j$ in the $\mu$Wave tier and the $\rm{mm}$Wave tier with the density $b_j \lambda_\mu$ and $b_j\lambda_{\rm{mm}}$, respectively.

\subsubsection{$\mu$Wave Tier} In the $\mu$Wave tier, the maximum-ratio transmission  beamforming is adopted at each SBS. All   channels undergo independent identically distributed (i.i.d.) quasi-static Rayleigh block fading. Without loss of generality, when a typical  $\mu$Wave UE  located at the origin $o$ requests the content $j$ from the {  associated $\mu$Wave BS $X_o$} that has cached this content, its received signal-to-interference-plus-noise ratio (SINR) is given by
\begin{align}\label{SINR_muWave}
{\mathrm{SINR}^{\mu}_j} =  \frac{{{P_\mu}{h^{\mu}_j}L\left( {\left| {{X^{\mu}_j }} \right|} \right)}}{{\mathcal{I}^\mu_j + \overline{\mathcal{I}}^\mu_j+ {\sigma^2_\mu}}},
\end{align}
where $P_\mu$ is the transmit power,   $h^{\mu}_j \sim \Gamma\left(N_\mu,1\right)$  is the \textcolor{black}{the equivalent small-scale fading} channel power gain  between the typical $\mu$Wave UE and its serving $\mu$Wave SBS, \textcolor{black}{where $\Gamma(k_1, k_2)$ denotes Gamma distribution, with a shape parameter $k_1$ and a scale parameter $k_2$}. The path loss is $L\left( {\left| {{X^{\mu}_j }} \right|} \right)=\beta_\mu {{\left( \left|X^{\mu}_j\right|\right)  }}^{ - {\alpha_\mu}}$ with the distance $\left| {{X^{\mu}_j }} \right|$, where $\beta_\mu$ is the frequency dependent constant parameter  and $\alpha_\mu$ is the path loss exponent. The $\sigma^2_\mu$ is the noise power at a $\mu$Wave UE. The inter-cell interference $\mathcal{I}^\mu_j $ and $\overline{\mathcal{I}}^\mu_j$ are  given by
\begin{equation} \label{Int_vec_exp_muWave}
\left\{\begin{aligned}
\mathcal{I}^\mu_j &=\sum\nolimits_{i \in {\Phi^\mu_j }\backslash X_o} {{P_{\mu}}
{h_{i,o}}L\left( {\left|X_{i,o}\right|} \right)},\\
\overline{\mathcal{I}}^\mu_j &=\sum\nolimits_{k \in {\overline{\Phi}^\mu_j}} {{P_{\mu}}
{h_{k,o}}L\left( \left|X_{k,o}\right| \right)}.
\end{aligned}\right.
\end{equation}
In \eqref{Int_vec_exp_muWave},   ${\Phi^\mu_j\backslash X_o}$ is the point process with density $b_j\lambda_\mu$  corresponding to the interfering SBSs that cache the content $j$,  and $\overline \Phi _j^\mu  = {\Phi ^\mu } - \Phi _j^\mu $ with density $\left(1-b_j\right)\lambda_\mu$ is the point process corresponding to the interfering SBSs that do not store the content $j$. The $h_{i,o}, h_{k,o} \sim \exp\left(1\right)$ are the interfering channel power gains that follow the exponential distribution, and $\left|X_{i,o}\right|, \left|X_{k,o}\right|$ denote the distances between the interfering SBSs and the typical UE.

\subsubsection{$\mathrm{mm}$Wave Tier} {\color{black}{In the mmWave tier, we assume that the directional beamforming is adopted at each mmWave SBS and small-scale fading is neglected, since small-scale fading has little change in received power as verified by
the practical mmWave channel measurements in \cite{TED2013IEEE_Access}. Note that the traditional small-scale fading distributions are invalid for mmWave modeling due to mmWave sparse scattering environment \cite{el2014spatially}.}} Unlike the conventional $\mu$Wave counterpart, $\mathrm{mm}$Wave transmissions are highly sensitive to the blockage. According to the average line-of-sight (LOS) model in \cite{T_Bai2014,J_Park_2016}, we consider that the $\mathrm{mm}$Wave link is LOS if the communication distance is less than $D_L$, and otherwise it is none-line-of-sight (NLOS). Moreover, the existing literature has confirmed that $\mathrm{mm}$Wave transmissions tend to be noise-limited and interference is weak \cite{T_Bai2014,Singh_2015}. Therefore, when a typical $\mathrm{mm}$Wave UE requests the content $j$ from the {associated} $\mathrm{mm}$Wave SBS that has cached this content, its received SINR is given by
\begin{align}\label{SINR_mmWave}
{\mathrm{SINR}^{\mathrm{mm}}_j} =  \frac{{{P_\mathrm{mm}}G_\mathrm{mm} L\left( {\left| {{Y^{\mathrm{mm}}_j }} \right|} \right)}}{{ {\sigma^2_\mathrm{mm}}}},
\end{align}
where $P_{\rm{mm}}$ is the transmit power of the $\mathrm{mm}$Wave SBS, $G_\mathrm{mm}$ is the main-lobe gain of using direction beamforming  {and equal to number of antenna elements \cite{venugopal2016device}}. The path loss is expressed as  $L\left( {\left| {{Y^{\mathrm{mm}}_j }} \right|} \right)=\beta_{\mathrm{mm}} {{\left( \left|Y^{\mathrm{mm}}_j\right|\right)  }}^{ - {\alpha}}$ with the distance $\left| {{Y^{\mathrm{mm}}_j }} \right|$ and frequency-dependent parameter $\beta_{\rm{mm}}$. The path loss exponent $\alpha=\alpha_\mathrm{L}$ when it is a LOS link and $\alpha=\alpha_\mathrm{N}$ when it is an NLOS link.  The $\sigma^2_\mathrm{mm}$ is the combined power of noise and weak interference \footnote{ $\rm{mm}$Wave in dense networks works in the noise-limited regime, since the high path loss impairs the interference, which could improve the signal directivity \cite{Singh_2015}. {\color{black}In contrast to the sub-6 GHz counterpart which is usually interference-limited, mmWave networks tend to be noise-limited when the BS density is not extremely dense, due to the narrow beam and blocking effects \cite{andrews2017modeling}.} For completeness, we also incorporate weak interference here.}.

\section{Successful Content Delivery Probability}
In this paper, SCDP is used as the performance indicator, which represents the probability that a content requested by a typical UE is both cached in the network and can be successfully transmitted to the UE. {\color{black}{We assume that each content has $\eta$ bits, and the delivery time needs to be less than $T$.}} {\color{black}{ By using the Law of total probability, the SCDP in the $\mu$Wave tier is calculated as}}
\begin{align}\label{SCDP_muWave}
\mathcal{P}_{\mathrm{SCD}}^{\mu}&=\sum\limits_{j{\rm{ = }}1}^J {{a_j}} {\rm{Pr}}\left( {{W_\mu }{{\log }_2}\left( {1{\rm{ + SINR}}_j^\mu } \right) \ge \frac{\eta }{T}} \right) \nonumber\\
&=\sum\limits_{j{\rm{ = }}1}^J {a_j} {\rm{Pr}}\left( \rm{SINR}_j^\mu >\varphi_\mu\right),
\end{align}
where {\color{black}{$W_{\mu}$ is the $\mu$Wave bandwidth allocated to a typical user (frequency-division multiple access (FDMA) is employed when multiple users are served by a SBS in this paper)}}, and $\varphi_\mu=2^{\frac{\eta}{W_\mu T}}-1$. Likewise, in the $\mathrm{mm}$Wave tier, the SCDP is calculated as
\begin{align}\label{SCDP_mmWave}
\mathcal{P}_{\mathrm{SCD}}^{\mathrm{mm}}&=\sum\limits_{j{\rm{ = }}1}^J {{a_j}} {\rm{Pr}}\left( {{W_{\mathrm{mm}} }{{\log }_2}\left( {1{\rm{ + SINR}}_j^{\mathrm{mm}} } \right) \ge \frac{\eta }{T}} \right) \nonumber\\
&=\sum\limits_{j{\rm{ = }}1}^J {a_j} {\rm{Pr}}\left( \rm{SINR}_j^{\mathrm{mm}} >\varphi_{\mathrm{mm}}\right),
\end{align}
where {\color{black}{$W_{\rm{mm}}$ is the $\rm{mm}$Wave bandwidth allocated to a typical user}}, and $\varphi_{\mathrm{mm}}=2^{\frac{\eta}{W_\mathrm{mm} T}}-1$.
The rest of this section is devoted to deriving the SCDPs in \eqref{SCDP_muWave} and \eqref{SCDP_mmWave}.

\subsection{$\mu$Wave Tier}
Based on \eqref{SINR_muWave} and \eqref{SCDP_muWave}, the SCDP in the $\mu$Wave tier can be derived and summarized below.
\begin{theorem}
In the cache-enabled $\mu$Wave tier, the SCDP is given by
\begin{align}\label{muWave_SCDP}
\mathcal{P}_{\mathrm{SCD}}^{\mu}&= \sum\limits_{j{\rm{ = }}1}^J {{a_j}} \mathcal{P}_{j,\mathrm{SCD}}^{\mu}\left(b_j \right),
\end{align}
where ${P}_{j,\mathrm{SCD}}^{\mu}\left(b_j \right)$ denotes the probability that the $j$-th request content is successfully delivered to the $\mu$Wave  UE by its serving SBS, and is expressed as
\begin{align}\label{muWave_STP}
{P}_{j,\mathrm{SCD}}^{\mu}\left(b_j \right)=\int_0^\infty  {P_{\operatorname{cov} }^\mu (x,b_j)}f_{\left| {{X}}_j^{\mu} \right|}(x) dx,
\end{align}
where ${P_{\operatorname{cov} }^\mu (x,b_j)}$ is given by \eqref{cov_pro_muWave} \textcolor{black}{at the top of the this page}, which represents the conditional coverage probability that the received SINR is larger than $\varphi _\mu$ given a typical communication distance $x$. $f_{\left| {{X}}_j^{\mu} \right|}(x)$ is the probability density function (PDF) of the distance $\left| {{X}}_j^{\mu} \right|$ between a typical $\mu$Wave UE and its nearest serving SBS that stores content $j$ , and is given by {\color{black}\cite{jo2012heterogeneous}}
\begin{align}\label{DRSP_PDF_M}
f_{\left| {{X}}_j^{\mu} \right|}(x) = 2\pi b_j {\lambda_\mu }x e^{ { - \pi b_j {\lambda_\mu } {x^2}} }.
\end{align}

\begin{figure*}[!t]
\normalsize
\begin{align}\label{cov_pro_muWave}
& {P_{\operatorname{cov} }^\mu (x,b_j)}=\sum\limits_{n = 0}^{{N_\mu} - 1} {\frac{{{{(x^{\alpha_\mu}) }^n}}}{{n!{{( - 1)}^n}}}\sum_{{\{t_q\}_{q=1}^n\in\Theta}} {\frac{{n!}}{{\prod\limits_{q = 1}^n {{t_q}!} {{(q!)}^{{t_q}}}}}} } \exp \bigg( { - \frac{{{\varphi _\mu}\sigma _\mu^2 x^{{\alpha_\mu}} }}{{{P_\mu} \beta_\mu }}}  { - 2\pi {b_j}{\lambda _\mu} \frac{{ {{{\varphi _\mu}  }} {x^{2 }}  }}{{{\alpha _\mu} - 2}}}
\nonumber\\
& {_2{F_1}\Big[ {1,\frac{{ - 2 + {\alpha_\mu}}}{{{\alpha _\mu}}},2 - \frac{2}{{{\alpha _\mu}}}, -{{{\varphi _\mu} }}  } \Big]} - \frac{{{2\pi ^2}}}{{{\alpha_\mu}}}(1 - {b_j}){\lambda _\mu}{({\varphi _\mu}x^{\alpha_\mu} )^{\frac{2}{{{\alpha _\mu}}}}}\csc\left( {\frac{{2\pi }}{{{\alpha _\mu}}}} \right)\bigg)  \prod\limits_{q = 1}^n {{{\left( {{\mathcal{T}^{\left( q \right)}}(x^{\alpha_\mu})} \right)}^{{t_q}}}},
\end{align}
where  $\Theta\triangleq \{\{t_q\}_{q=1}^n|\sum\limits_{q=1}^{n}q \cdot t_q =n, \mbox{$t_q$ is an integer}, \forall n\}$, $\csc\left(\cdot\right)$ is the Cosecant trigonometry function, and
\begin{align}
\label{Lambda_2}
 {\mathcal{T}_j^{(1)}}(x^{\alpha_\mu}) = &  - \frac{{{\varphi_\mu}\sigma_\mu^2}}{{{P_\mu}\beta }} - 2\pi {{b_j}{\lambda_\mu}}{x^{2 - {\alpha_\mu}}}{\varphi_\mu}  \frac{{{\alpha_\mu} - 2 + 2\left( {1 + {\varphi_\mu}} \right){}_2{F_1}\left[ {1,\frac{{{\alpha_\mu} - 2}}{{{\alpha_\mu}}},2 - \frac{2}{{{\alpha_\mu}}}, - {\varphi_\mu}} \right]}}{{\left( {1 + {\varphi_\mu}} \right)\left( {{\alpha_\mu} - 1} \right){\alpha_\mu}}} \nonumber \\
& - 4{\pi ^2}(1 - {b_j}){\lambda_\mu}{\left( {\varphi_\mu^{\alpha_\mu} x } \right)^{2 - {\alpha_\mu}}}
\csc \left( {\frac{{2\pi }}{{{\alpha_\mu}}}} \right),
\end{align}
\begin{align}
{\mathcal{T}_j^{(q)}}(x^{\alpha_\mu}) = & 2\pi{b_j}{\lambda_\mu}q!{\left( { - 1} \right)^q}{x^{ -({{2 + {\alpha_\mu}}})(1 + q)}}\varphi_\mu^{ - q(1 + q)}  \frac{{{}_2{F_1}\left[ {1 + q,\frac{{2 + {\alpha_\mu}}}{{{\alpha_\mu}}},2 + \frac{2}{{{\alpha_\mu}}}, - \frac{1}{{{\varphi_\mu}}}} \right]}}{{2 + {\alpha_\mu}}} +2\pi (1 - {b_j}){\lambda_\mu}q!{\left( { - 1} \right)^q}\nonumber \\
  &\times {(x^{\alpha_\mu})^{ - q + \frac{2}{{{\alpha_\mu}}}}} \varphi_\mu^{\frac{2}{{{\alpha_\mu}}}}\frac{{\Gamma \left( {q - \frac{2}{{{\alpha_\mu}}}} \right)\Gamma \left( {\frac{{2 + {\alpha_\mu}}}{{{\alpha_\mu}}}} \right)}}{{{\alpha_\mu}\Gamma \left( {1 + q} \right)}}, \;\; q>1.
\end{align}
\hrulefill
\end{figure*}
\end{theorem}
\begin{proof}
Please see Appendix A.
\end{proof}
Note that  ${P}_{j,\mathrm{SCD}}^{\mu}\left(b_j \right)$ becomes the probability of successful transmission from the serving SBS to the typical user when $b_j$=1 in traditional  $\mu$Wave networks without caching. {\color{black}{We see that the SCDP expression for multi-antenna systems is much complicated, compared to the closed-form expression for single-antenna systems in \cite{wen2016cache}.}}

\subsection{$\mathrm{mm}$Wave Tier}
Based on \eqref{SINR_mmWave} and \eqref{SCDP_mmWave}, the SCDP in the $\mathrm{mm}$Wave tier can be derived and summarized below.
\begin{theorem}
In the cache-enabled $\mathrm{mm}$Wave tier, the SCDP is given by
\begin{align}\label{mmWave_SCDP}
\mathcal{P}_{\mathrm{SCD}}^{\rm{mm}}=\sum\limits_{j{\rm{ = }}1}^J {{a_j}} \mathcal{P}_{j,\mathrm{SCD}}^{\rm{mm,L}}(b_j)+\sum\limits_{j{\rm{ = }}1}^J {{a_j}} \mathcal{P}_{j,\mathrm{SCD}}^{\rm{mm,N}}(b_j),
\end{align}
where $\mathcal{P}_{j,\mathrm{SCD}}^{\rm{mm,L}}(b_j)$ and $\mathcal{P}_{j,\mathrm{SCD}}^{\rm{mm,N}}(b_j)$ denote that probabilities that the content $j$ is successfully delivered when the $\mathrm{mm}$Wave UE is connected to its serving $\mathrm{mm}$Wave SBS via LOS link and NLOS link,   and are given by
\begin{align}\label{P_m_los_closeform}
\mathcal{P}_{j,\mathrm{SCD}}^{\rm{mm,L}}(b_j) = 1- e^{ - \left( \min{\left( D_{\text{L}},{d_{\rm{L}}} \right)} \right)^2 \pi b_j \lambda_{\rm{mm}} },
\end{align}
and
\begin{align}\label{P_m_nlos_closeform}
 \mathcal{P}_{j,\mathrm{SCD}}^{\rm{mm,N}}(b_j)= e^{ - {D_\text{L}^2}\pi b_j \lambda_{\rm{mm}} }- e^{ - \left(\max{\left( {D_{\text{L}},{d_{\rm{N}}}}\right)}\right)^2 \pi b_j {\lambda_{\rm{mm}}}},
\end{align}
respectively, where $ d_\text{L}= {\left( {\frac{{P_\mathrm{mm} G_\mathrm{mm}\beta_\mathrm{mm} }}{{{\varphi _\mathrm{mm}}\sigma _\mathrm{mm}^2}}} \right)}^{ \frac{1}{\alpha _{\text{L}}}} $ and $ d_\text{N}= {\left( {\frac{{P_\mathrm{mm} G_\mathrm{mm}\beta_\mathrm{mm} }}{{{\varphi_\mathrm{mm}}\sigma _\mathrm{mm}^2}}} \right)}^{ \frac{1}{\alpha _{\text{N}}}} $.

\end{theorem}

\begin{proof}
Please see Appendix B.
\end{proof}

\section{Optimization of Probabilistic Content Placement}

In this section,  we aim to maximize the SCDP by optimizing the probabilistic content placement $\{b_j\}$. The main difficulty is that the SCDP expressions \eqref{muWave_SCDP} and \eqref{mmWave_SCDP} do not have a closed form  for the multi-antenna case and whether they are concave with regard to  $\{b_j\}$ is unknown, which is much more challenging than the single-antenna SBS case studied in \cite{wen2016cache}. {\color{black}{Therefore, the optimal content placement problem for the multi-antenna case is distinct. To tackle this new problem, here we propose two algorithms,}} the first one is developed  based on the  CEO method that can achieve near-optimal performance, and the other two-stair scheme is based on the combination of MPC and CD content placement schemes with reduced complexity.

\subsection{The Near-Optimal CCEO Algorithm}
\textcolor{black}{ The optimal caching placement probability in the multi-antenna case is hard to achieve, so we introduce CEO to resolve the difficulty of maximizing the SCDP by optimizing the probabilistic content placement. CEO is an adaptive variance algorithm for estimating probabilities of rare events. The rationale of the CEO algorithm is to first associate with each optimization problem a rare event estimation problem, and then to tackle this estimation problem efficiently by an adaptive algorithm. The  outcome of this algorithm is the construction of a random sequence of solutions which converges probabilistically to the optimal or near-optimal solution \cite{CEO,botev2013cross}. The CEO method involves   two iterative steps. The first one is to generate samples of random data   according to a specified random (normally Gaussian) distribution. And the second step updates the parameters of the random distribution, based on the sample data  to produce better samples in the next iteration.}
\textcolor{black}{
The CEO algorithm has been successfully applied to a wide range of difficult optimization tasks such as traveling salesman problem and antenna selection problem in multi-antenna communications \cite{CEO-antenna}. It has shown superior performance in solving complex  optimization problems compared to commonly used simulated annealing (SA) and genetic algorithm (GA) \cite{caserta2009cross} that are based on random search.}

\textcolor{black}{
 The original principle of the CEO algorithm was proposed for unconstrained optimization. To deal with the constraints on the probabilities $\{b_j\}$ and the content capacity constraint, we propose a CCEO algorithm as shown in Algorithm \ref{alg:A}. In the proposed CCEO algorithm, we force the randomly generated samples to be within the feasible set $\{b_j|0 \leq b_j \leq 1, \; \forall j\}$ in the Project step. To satisfy the constraint of $\sum_j^J b_j\le M$, we introduce a penalty function $H\left(\sum_j^J b_j- M\right)$ to the original objective function in the Modification step, where $H$ is a large positive number  that represents the parameter for the penalty function. The dynamic Smoothing step will prevent the result from converging to a sub-optimal solution. It can be seen that at each iteration, the main computation is to evaluate the objective functions  for $N_s$ times and no gradient needs to be calculated, so the complexity is moderate and can be further controlled to achieve a complexity-convergence tradeoff.}
\textcolor{black}{
\begin{algorithm}
\caption{Constrained Cross-Entropy Optimization (CCEO) Algorithm}
\label{alg:A}
\begin{algorithmic}
\STATE  \textbf{Initialization}: Randomly initialize the parameters of Gaussian distribution $\mathcal{N}(\mu_{j,t=0}, \sigma^2_{j, t=0})$ where $t = 0$ is the iteration index. Set sample number $N_{s}$, the number of selected samples $N^{elite} \ll N_s$   the stopping threshold $\epsilon$ and a large positive number $H$ as the parameter for the penalty function.
\REPEAT
\STATE  \textbf{Sampling:}   Generate $N_{s}$ random samples $\bm{b} = \{\bm{b}_{1},\bm{b}_{2},..\bm{b}_{j},...\bm{b}_{N_s} \}$ from the $ \mathcal{N}(\mu_{t}, \sigma^2_{t})$ distribution.
\STATE  \textbf{Projection:}   Project the samples onto the feasible set $\{b_j|0 \leq b_j \leq 1, \; \forall j\}$, i.e., $\bm{b}  = \min(\max(\bm{b}, 0),1)$.
\STATE\textbf{Modification:} We modify the objective function to the following:
\begin{align}
\widehat{\mathcal{P}}_{\text{SCD}}(\bm{b})=\mathcal{P}_{\text{SCD}}(\bm{b})-  H \max(\sum\nolimits_{j = 1}^J {{b_j}} -M,0),
\end{align}
where $\mathcal{P}_{\text{SCD}}(\bm{b})$ is the original objective function in \eqref{muWave_SCDP} and \eqref{mmWave_SCDP} for $\mu$Wave and mmWave, respectively.
\STATE\textbf{Selection:}  Evaluate $\widehat{\mathcal{P}}_{\text{SCD}}(\bm{b})$ for $N_s$ samples $\bm{b}$. Let $\mathcal I$ be the indices of the $N^{elite}$ selected best performing samples with $\widehat{\mathcal{P}}_{\text{SCD}}(\bm{b})$.
\STATE \textbf{Updating:} for all $j \in \mathcal F$, calculate the sampling mean and variance:
\begin{align}
{\widetilde \mu _{ij}} = \sum\limits_{i \in \mathcal I} {{b_{ij}}/{N^{elite}}}
\end{align}
\begin{align}
{\widetilde {{\sigma }}^2_{ij}} = \sum\limits_{i \in \mathcal I} {{{\left( {{b_{ij}} - {{\widetilde \mu }_{ij}}} \right)}^2}/{N^{elite}}}.
\end{align}
\STATE \textbf{Smoothing:}
The Gaussian distribution parameters are updated as follows,
\begin{align}\label{mu}
{\bm{\mu} _t} = \iota {\widetilde {\bm{\mu}} _t} + (1 - \iota ){\bm{\mu} _{t - 1}},
\end{align}
\begin{align}\label{sigma}
{\bm{\sigma} }^2_t = \beta_t {\widetilde {{\bm{\sigma} }}^2_t} + (1 - \beta_t )\bm{\sigma}^2_{t - 1}.
\end{align}
In particular, $\alpha$ is a fixed  smoothing parameter ($0.5 \le \alpha \le 0.9$) while $\beta_t$ is a dynamic smoothing parameter given by
\begin{equation}
    \beta_t = \beta - \beta \left(1-\frac{1}{t}\right)^q,
\end{equation}
where $\beta$ is a fixed  smoothing parameter ($0.8\le \beta \le 0.99$), and $q$ is an integer with a typical value between 5 and 10.
\STATE \textbf{Increment:} $t=t+1$.
\UNTIL{A convergence criterion is satisfied, e.g., $\mathop{\max}\limits_{j \in \mathcal F} (\bm{\sigma}^2_{t}) < \epsilon$}
\STATE \textbf{Output:}  The optimal caching probability is $\bm{b}^{\ast}={\bm{\mu} _t}$.
\end{algorithmic}
\end{algorithm}}

\begin{figure}
\centering
\includegraphics[width=3.6in, height=3 in]{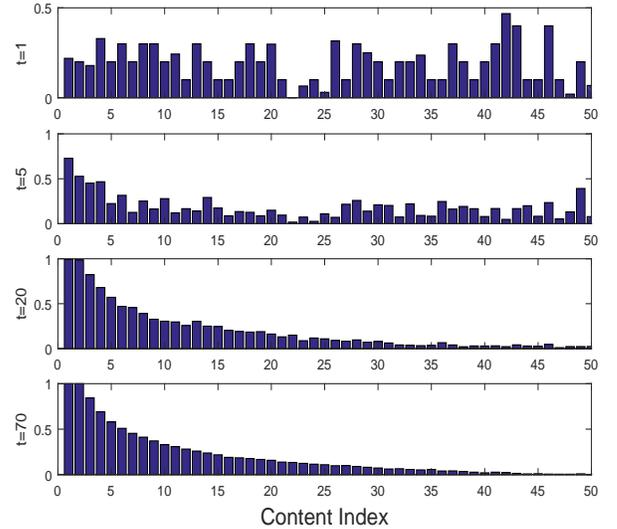}
\caption{\label{Fig_Traning}Evolution of the content placement probabilities in the CCEO algorithm with parameters $\gamma=1.5$, $J=50$, $M=10$.}
\end{figure}

In Fig. 2, we provide an example of the iterative results of content placement probabilities with iteration indices $t=1$, $t=5$, $t=20$, and $t=70$. In this example, the algorithm converges when $t=70$. Each sub-figure presents the resulting mean value of ${\bm{\mu}}_t$ at the end of iteration $t$, and it will help to generate random samples in next iteration.  We can  observe that when $t=20$, the caching placement probability is quite close to the converged solution, which could significantly reduce  the complexity. {\color{black}{Overall the CEO algorithm converges fast and is an efficient method to find the near-optimal SCDP result, and the complexity of the CEO algorithm is $\mathcal{O}\left(n^{3}\right)$~\cite{rubinstein1999cross}.}} It is also noted that the top ranked   contents are cached with probability $b_j=1$, while to make effective use of the rest caching space, caching diversity is more important. Based on this observation, we design a low-complexity heuristic scheme in the next subsections.

\subsection{Two-Stair Scheme for the $\mu$Wave Tier}

\begin{figure}[h]
\centering
\includegraphics[width=1.5 in]{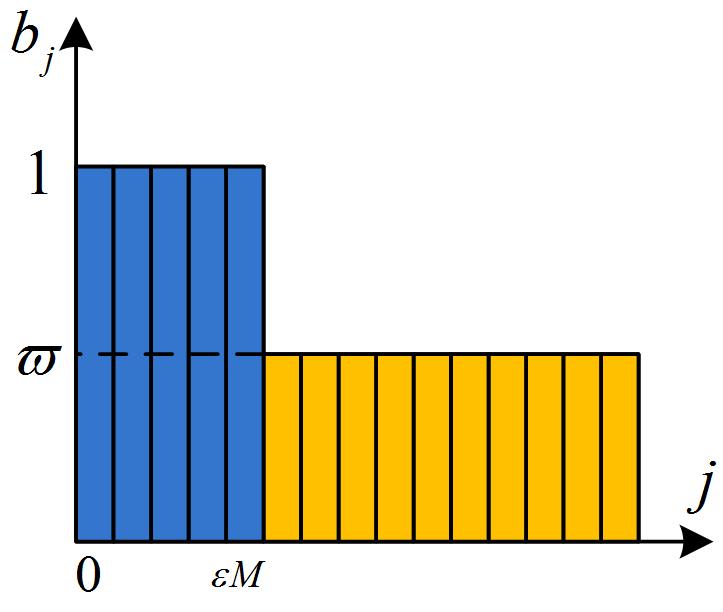}
\caption{\label{figure:bj}\textcolor{black}{Two-Stair probabilistic content placement strategies}.}
\label{sys}
\end{figure}
\textcolor{black}{
To further reduce the complexity of the optimization, we devise a simple two-stair (TS) scheme, when the content popularity is modeled as the {Zipf} distribution \cite{chen2016cooperative,Zipf,Blaszczyszyn_2015} based on empirical studies, which is given by
\begin{align}
{a}_j = {j^{-\gamma} }/\sum\nolimits_{m= 1}^J {{m^{-\gamma} }} ,
\end{align}
where $\gamma$ is the Zipf exponent that represents the popularity skewness.}

In the TS scheme, a fraction of caching space $\varepsilon M$ ($0\leq \varepsilon \leq 1$)  at a SBS is allocated to store the most popular contents which is called the MPC region. The remaining cache space is allocated to randomly store the contents with certain probabilities and is called the CD region. As illustrated in Fig. \ref{figure:bj}, in the `Two-Stair' caching scheme, the contents in the CD region are cached with a common probability $\varpi$. The rest of the contents are not cached and must be fetched through the backhaul links. These content placement schemes will be studied in detail in the rest of this section.

In this scheme, the content placement probabilities  $\left\{b_j\right\}$ need to satisfy the following conditions:
 \begin{align}\label{CD_prob}
\left\{ \begin{gathered}
 b_{1}  =  \ldots  =  b_{\left\lfloor {\varepsilon M} \right\rfloor    }  = 1 , \hfill\\
  {b_{\left\lfloor {\varepsilon M} \right\rfloor  + 1}} =  \ldots  = {b_{\left\lfloor {\varepsilon M} \right\rfloor  + \left\lfloor {\frac{{M - \left\lfloor {\varepsilon M} \right\rfloor }}{\varpi }} \right\rfloor }} = \varpi , \hfill \\
  {b_{\left\lfloor {\varepsilon M} \right\rfloor  + \left\lfloor {\frac{{M - \left\lfloor {\varepsilon M} \right\rfloor }}{\varpi }} \right\rfloor  + 1}} =  \ldots  = {b_J} = 0, \hfill
\end{gathered}  \right.
\end{align}
which are characterized by two variables $\varepsilon$ and $\varpi$,
where $\varpi$ denotes the common probability value that content $j$   in the CD region is stored at a SBS.
 %and $\mathbf{1}\left(A\right)$ is the indicator function that returns one if the condition A is satisfied.
%Fig. 1(a) shows an example of combined MPC and CD caching with $J=7$, $M=4$ $\varepsilon=0.5$ and $\varpi=0.4$, where the caching probability values $\left\{b_j\right\} $ of the contents fills the cache memory of each BS. We select a uniformly random number within $[0,1]$ and draw a vertical dashed line. Then, the contents are chosen to be cached if their memory blocks are intersected with this  line. Therefore, in Fig. 1(a), the most popular contents $\left\{ f_1, f_2 \right\}$ are cached at each BS based on the MPC caching, and the contents $\left\{ f_5, f_7 \right\}$  are selected to be cached at a BS (random number 0.9) based on the CD caching.

%In Fig. \ref{figure:bj}, we were shown two MPC and CD caching placement, Fig. \ref{fig:subfig:a} is created by Two-stair (TS) structure which guarantees save the high popularly rank content in each SBS then storage the rest part with the same probability. Fig. \ref{fig:subfig:b} is created by Multi-stair linear decrease (MT-LD) structure which storage the CD part with multi-stair linear decrease probability.

As such, the $\mu$Wave SCDP \eqref{muWave_SCDP} can be   expressed as
\begin{align}\label{SCDP_prob_muWave}
\mathcal{P}_{\mathrm{SCD}}^{\mu}&= \sum\limits_{j{\rm{ = }}1}^{\left\lfloor {\varepsilon M} \right\rfloor} {{a_j}} \mathcal{P}_{j,\mathrm{SCD}}^{\mu}\left(1 \right) \nonumber\\
&+\sum\limits_{j{\rm{ = }}\left\lfloor {\varepsilon M} \right\rfloor+1}^{\left\lfloor {\varepsilon M} \right\rfloor+\left\lfloor \frac{{ M- \left\lfloor {\varepsilon M} \right\rfloor}}{\varpi} \right\rfloor} {{a_j}} \mathcal{P}_{j,\mathrm{SCD}}^{\mu}\left(\varpi \right).
\end{align}
It is seen in \eqref{SCDP_prob_muWave} that contents $\{1,\cdots,\left\lfloor {\varepsilon M} \right\rfloor\}$ have the same SCDP $\mathcal{P}_{j,\mathrm{SCD}}^{\mu}\left(1\right)$, and contents $\left\{\left\lfloor {\varepsilon M} \right\rfloor+1,\cdots,\left\lfloor {\varepsilon M} \right\rfloor+\left\lfloor \frac{{ M- \left\lfloor {\varepsilon M} \right\rfloor}}{\varpi} \right\rfloor\right\}$ have the same SCDP $\mathcal{P}_{j,\mathrm{SCD}}^{\mu}\left(\varpi \right)$.
Our aim is to maximize the overall SCDP, and the problem is formulated as
\begin{align}\label{Problem_formulation_muWave}
 &\mathop{\max }\limits_{\varepsilon,\varpi} \mathcal{P}_{\mathrm{SCD}}^{\mu} \mbox{~~in \eqref{SCDP_prob_muWave}} \quad  \nonumber \\
&\mathop{\rm{s.t.}} \;\;~\mathrm{C1:}~0 \leq \varepsilon \leq 1, \nonumber \\
&\;\;\qquad \mathrm{C2:}~ 0 \leq \varpi  \leq 1, \nonumber \\ %\frac{{ M-\left\lfloor {\varepsilon M} \right\rfloor}} {J-\left\lfloor {\varepsilon M} \right\rfloor}
&\;\;\qquad \mathrm{C3:}~\mathbf{1}\left(\varepsilon=1\right) \varpi=0,
\end{align}
where $\mathbf{1}\left(A\right)$ is the indicator function that returns one if the condition A is satisfied.
The convexity of the problem \eqref{Problem_formulation_muWave} is unknown, and finding its global optimal solution is challenging. To obtain an efficient caching placement solution, we first use the following approximations \cite{taghizadeh2013distributed}
\begin{align}\label{approx_populairty}
& \sum\limits_{j{\rm{ = }}1}^{\left\lfloor {\varepsilon M} \right\rfloor} {{a_j}} \approx
\frac{{{{\left( {\varepsilon M} \right)}^{1 - \gamma }} - 1}}{{{J^{1 - \gamma }} - 1}},
\end{align}
\begin{align}\label{approx_populairty_1}
&\sum\limits_{j{\rm{ = }}\left\lfloor {\varepsilon M} \right\rfloor+1}^{\left\lfloor {\varepsilon M} \right\rfloor+\left\lfloor \frac{{ M- \left\lfloor {\varepsilon M} \right\rfloor}}{\varpi} \right\rfloor} {{a_j}} \approx
  \frac{{{{\left( {\varepsilon M + \frac{{M(1 - \varepsilon )}}{\varpi }} \right)}^{1 - \gamma }} - 1}}{{{J^{1 - \gamma }} - 1}} - \frac{{{{\left( {\varepsilon M} \right)}^{1 - \gamma }} - 1}}{{{J^{1 - \gamma }} - 1}} \nonumber \\
&   = \frac{{{M^{1 - \gamma }}}}{{{J^{1 - \gamma }} - 1}}\left[ {{{\left( {\varepsilon  + \frac{{(1 - \varepsilon )}}{\varpi }} \right)}^{1 - \gamma }} - {\varepsilon ^{1 - \gamma }}} \right] ,
\end{align}
respectively, based on the fact that for Zipf popularity with $0<\gamma$, $\gamma \neq 1$ and $M \ll J$, we have $\sum\nolimits_{j = 1}^M {{j^{ - \gamma }}} /\sum\nolimits_{m = 1}^J {{m^{ - \gamma }}}  \approx \left( {{M^{1 - \gamma }}-1} \right)/\left( {{J^{1 - \gamma }}-1} \right)$~\cite{taghizadeh2013distributed}. Therefore, the objective function of \eqref{SCDP_prob_muWave} can be  approximated as
\begin{align}\label{obj_approx_muWave}
& {\widetilde{\mathcal{P}}}_{\mathrm{SCD}}^{\mu}\approx
  \mathcal{P}_{j,{\text{SCD}}}^\mu \left( 1 \right)\frac{{{M^{1 - \gamma }}}}{{{J^{1 - \gamma }} - 1}}{\varepsilon ^{1 - \gamma }} - \frac{{\mathcal{P}_{j,{\text{SCD}}}^\mu \left( 1 \right)}}{{{J^{1 - \gamma }} - 1}} \nonumber \\
  & + \mathcal{P}_{j,{\text{SCD}}}^\mu \left( \varpi  \right)\frac{{{M^{1 - \gamma }}}}{{{J^{1 - \gamma }} - 1}}\left[ {{{\left( {\varepsilon  + \frac{{(1 - \varepsilon )}}{\varpi }} \right)}^{1 - \gamma }} - {\varepsilon ^{1 - \gamma }}} \right]  .
\end{align}
Note that for the special case of MPC caching, i.e., $\varepsilon=1,\varpi=0$, the above reduces to ${\widetilde{\mathcal{P}}}_{\mathrm{SCD}}^{\mu}\approx \mathcal{P}_{j,\mathrm{SCD}}^{\mu}\left(1\right) \frac{{{M^{1 - \gamma }} - 1}}{{{J^{1 - \gamma }} - 1}}$.

Then the problem \eqref{Problem_formulation_muWave} can be approximated as
\begin{align}\label{Problem_formulation_muWave_relaxation}
 &\mathop{\max }\limits_{\varepsilon,\varpi} ~{\widetilde{\mathcal{P}}}_{\mathrm{SCD}}^{\mu} \quad  \nonumber \\
&\mathop{\rm{s.t.}}  ~~\mathrm{C1} \mbox{ --  }\mathrm{C3}.
\end{align}
Because $\varepsilon$ and $\varpi$ are coupled in the objective function of \eqref{Problem_formulation_muWave_relaxation}, we use a decomposition approach to solve this problem. Since ${{M^{1 - \gamma }}}$ is always positive, given $\varpi$, the optimal $\varepsilon$ is obtained by solving the following equivalent sub-problem:
\begin{align}\label{Problem_formulation_muWave_sub}
&\mathop{\max }\limits_{0 \leq \varepsilon \leq 1 } \frac{1}{{{J^{1 - \gamma }} - 1}}\left[ {\left( {\ell _o^\mu  - 1} \right){\varepsilon ^{1 - \gamma }} + {{\left( {\varepsilon  + \frac{{(1 - \varepsilon )}}{\varpi }} \right)}^{1 - \gamma }} - \ell _o^\mu } \right]
%   \quad  \nonumber \\
%&\mathop{\rm{s.t.}} \;\;~\mathrm{C1:}~0 \leq \varepsilon \leq 1, \nonumber \\
%&\;\;\qquad \mathrm{C2:}~\frac{{ M-{\varepsilon M} }} {J- {\varepsilon M} } \leq \varpi  \leq 1.
\end{align}
where $\ell_o^{\mu}=\frac{\mathcal{P}_{j,\mathrm{SCD}}^{\mu}\left(1 \right)}{\mathcal{P}_{j,\mathrm{SCD}}^{\mu}\left(\varpi\right)}\ge 1$ is independent of $\varepsilon$. Thus, we have the following theorem:
\begin{theorem}
The optimal solution of the problem \eqref{Problem_formulation_muWave_sub} is given by
\begin{align}\label{optimal_P2_1}
\varepsilon^*=\min(\max({\varepsilon_o},0),1),
\end{align}
where $\varepsilon_o={\left( {\left( {{{\left( {\frac{{{\ell _o^{\mu}} - 1}}
{{{\varpi ^{ - 1}} - 1}}} \right)}^{ - 1/\gamma }} - 1} \right)\varpi  + 1} \right)^{ - 1}}$.
\end{theorem}
\begin{proof}
Please see Appendix C. For $\varepsilon_o$ to be in the range of $[0,1]$, $\varpi$ should satisfy $0\le \varpi\le \frac{1}{\ell _o^{\mu}}.$
\end{proof}
%It is indicated from \eqref{optimal_P2_1} that only when $0 < \varpi < \frac{1}{\ell_o^{\mu}}$ (i.e., $\varepsilon^*\left(\varpi\right)<1$), the less popular contents are possible to be cached at BSs to exploit the content diversity, otherwise, each BS only caches $M$ most popular contents following MPC caching. Therefore, based on \textbf{Theorem 3} and the problem \eqref{Problem_formulation_muWave_relaxation}, we first need to solve the following sub-problem:
Consequently, the problem \eqref{Problem_formulation_muWave_relaxation} reduces to the following optimization problem about $\varpi$ only:
\begin{align}\label{Problem_formulation_muWave_sub2}
&\mathop{\max }\limits_{\varepsilon=\varepsilon_o(\varpi),0 \le \varpi \le \frac{1}{\ell _o^{\mu}}}  {\widetilde{\mathcal{P}}}_{\mathrm{SCD}}^{\mu}. %\quad  \nonumber \\
%&\mathop{\rm{s.t.}} \;\;~\mathrm{C1:}~0 < \varpi < \frac{1}{\ell_o}.
\end{align}
Since the problem \eqref{Problem_formulation_muWave_sub2} is non-convex, we propose to use Newton's method to solve it, which is shown in the Appendix D. {\color{black}{Note that the Newton's method converges faster than the Karush-Kuhn-Tucker (KKT) method and the gradient-based method~\cite{yu2013multicell}.}} Suppose the obtained solution is $\hat{\varpi}$, then the optimal $\hat{\varpi}^*$ is $\min(\max({\hat{\varpi}},0),1)$, and the optimal $\varepsilon^*$ can be obtained from \eqref{optimal_P2_1}.
%
%After obtaining the solution $(\hat{\varepsilon},\hat{\varpi})$  via the proposed algorithm, let ${\hat{\mathcal{P}}}_{\mathrm{SCD}}^{\mu}$ denote the SCDP value by substituting $(\hat{\varepsilon},\hat{\varpi})$ into \eqref{obj_approx_muWave}, and the solution of the problem \eqref{Problem_formulation_muWave_relaxation} is obtained as
%\begin{align}\label{optimal_P2}
%\left(\varepsilon^*,\varpi^*\right)=\left\{ \begin{array}{l}
%\left(\hat{\varepsilon},\hat{\varpi}\right),\quad {\hat{\mathcal{P}}}_{\mathrm{SCD}}^{\mu}> \mathcal{P}_{j,\mathrm{SCD}}^{\mu}\left(1 \right) \left( {\frac{M}{J}}\right)^{1 - \gamma }\\
%\left(1,0\right), \quad {\hat{\mathcal{P}}}_{\mathrm{SCD}}^{\mu}\leq  \mathcal{P}_{j,\mathrm{SCD}}^{\mu}\left(1 \right) \left( {\frac{M}{J}}\right)^{1 - \gamma }
%\end{array} \right..
%\end{align}
%It is shown in \eqref{optimal_P2} that the SCDP achieved by the proposed caching placement is greater than that with the MPC caching.
%\begin{proof}
%Please see Appendix D.
%\end{proof}

\subsection{Two-Stair Scheme for the $\mathrm{mm}$Wave Tier}
Similar to the $\mu$Wave case, the SCDP of the $\mathrm{mm}$Wave tier can be approximated by
 \begin{align}\label{obj_mmWave_caching}
&\mathcal{P}_{\mathrm{SCD}}^{\mathrm{mm}}\approx\ \left( {\mathcal{P}_{j,{\mathrm{SCD}}}^{{\mathrm{mm}},{\text{L}}}(1) + \mathcal{P}_{j,{\mathrm{SCD}}}^{{\mathrm{mm}},{\text{N}}}(1)} \right)\frac{{ {{(\varepsilon M)}^{1 - \gamma }} - 1}}{{{J^{1 - \gamma }} - 1}} \nonumber \\
&  + \left( {\mathcal{P}_{j,{\text{SCD}}}^{{\text{mm}},{\text{L}}}(\varpi ) + \mathcal{P}_{j,{\mathrm{SCD}}}^{{\mathrm{mm}},{\text{N}}}(\varpi )} \right)\ \times \nonumber \\
&\frac{{{M^{1 - \gamma }}}}{{{J^{1 - \gamma }} - 1}}\left[ {{{\left( {\varepsilon  + \frac{{(1 - \varepsilon )}}{\varpi }} \right)}^{1 - \gamma }} - {\varepsilon ^{1 - \gamma }}} \right].
\end{align}
Then the optimal two-stair content caching  can be found obtained by solving the following problem:
\begin{align}\label{Problem_reformul_umWave}
 &\mathop{\max }\limits_{\varepsilon,\varpi} \mathcal{P}_{\mathrm{SCD}}^{\mathrm{mm}} \mbox{~in \eqref{obj_mmWave_caching}} \quad  \nonumber \\
&\mathop{\rm{s.t.}} \;\;~\mathrm{C1} \mbox{ -- }\mathrm{C3}.
\end{align}
The problem \eqref{Problem_reformul_umWave} can be efficiently solved by following the decomposition approach. Given $\varpi$, the optimal $\varepsilon$ is obtained by solving the following equivalent sub-problem:
\begin{align}\label{Problem_formulation_mmWave_sub2}
&\mathop{\max }\limits_{0 \leq \varepsilon \leq 1 } \frac{{\left( {\ell _o^{{\text{mm}}} - 1} \right){\varepsilon ^{1 - \gamma }} + {{\left( {\varepsilon  + \frac{{(1 - \varepsilon )}}{\varpi }} \right)}^{1 - \gamma }} - \ell _o^{{\text{mm}}}}}{{{J^{1 - \gamma }} - 1}},
\end{align}
where $\ell_o^{\rm{mm}}=\frac{\left(\mathcal{P}_{j,\mathrm{SCD}}^{\rm{mm,L}}(1)+
\mathcal{P}_{j,\mathrm{SCD}}^{\rm{mm,N}}(1)\right)}{\left(\mathcal{P}_{j,\mathrm{SCD}}^{\rm{mm,L}}(\varpi)+
\mathcal{P}_{j,\mathrm{SCD}}^{\rm{mm,N}}(\varpi)\right)}$. The rest procedures follow the same approach in the section IV-A, except that the derivation of the search direction to solve the optimal $\varpi^*$, which is provided Appendix E.

\section{Results and Discussions}
In this section, the performance of the proposed caching schemes are evaluated by presenting numerical results. Performance comparison between cache-enabled $\mu$Wave and $\mathrm{mm}$Wave systems is also highlighted. The system parameters are shown in Table \ref{tab:passloss}, unless otherwise specified. 1 GHz and 60 GHz are chosen for the $\mu$Wave   and  $\mathrm{mm}$Wave frequency bands, respectively.
\begin{table}\footnotesize
\centering
\caption{Parameter Values.}\label{tab:passloss}
\begin{tabular}{c|c}
\hline
\textbf{Parameters}& \textbf{Values}\\
\hline \hline
Number of Antenna in $\mu$Wave-SBS($N_{\mu}$)& 2 \\ \hline
Main-lobe Array Gain in $\mathrm{mm}$Wave-SBS ($G_\text{mm}$)& 2 \\ \hline
LOS region ($D_{\rm{L}}$) & 15 m\\ \hline
Transmit power of each $\mu$Wave-SBS $P_\mu$& 20 dBm\\ \hline
Transmit power of each $\mathrm{mm}$Wave-SBS $P_\text{mm}$& 20 dBm\\ \hline
SBS's density for $\mu$Wave and $\mathrm{mm}$Wave & $\lambda_{\mu}$,$\lambda_{\mathrm{mm}}$=600/km$^2$\\ \hline
Path loss exponent $f_c$=1 GHz&$\alpha_\mu$=2.5\\ \hline
%Path loss exponent $f_c$=28 GHz&$\alpha_{\text L}$=2,$\alpha_{\text N}$=3\\ \hline
%Path loss exponent $f_c$=73 GHz&$\alpha_{\text L}$=2,$\alpha_{\text N}$=3.4\\ \hline
%Path loss exponent $f_c$=38 GHz \cite{rappaport201238}&$\alpha_{\text L}$=2,$\alpha_{\text N}$=3.71\\ \hline
Path loss exponent $f_c$=60 GHz \cite{rappaport201238}&$\alpha_{\text L}$=2.25,$\alpha_{\text N}$=3.76\\ \hline
Bit rate of each \textcolor{black}{content} ($\eta/T$) &$4 \times 10^5 $bit/s\\ \hline
Available bandwidth in {$\mu$Wave} ($W_{\mu}$) &10 MHz\\ \hline
Available bandwidth in $\mathrm{mm}$Wave ($W_{\text{mm}}$) &1 GHz\\ \hline
SBS cache capacity ($M$)& \textcolor{black}{10}\\ \hline
Content library size ($J$) & \textcolor{black}{100}\\ \hline
Zipf exponent ($\gamma$)& 0 $\sim$ 2\\ \hline
\end{tabular}
\end{table}

\begin{figure}
\centering
\includegraphics[width=3.6in, height=3 in]{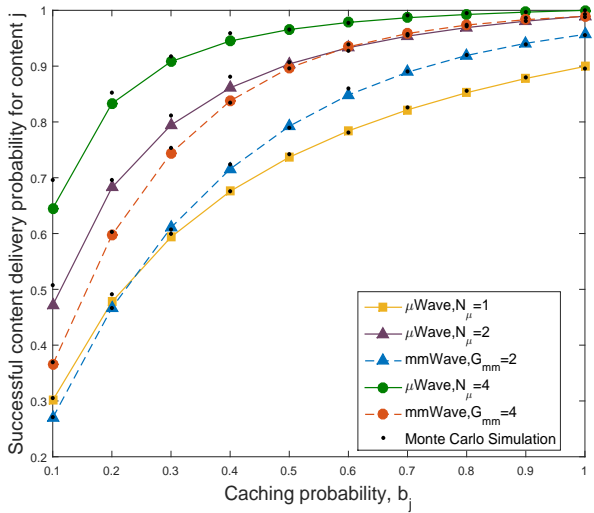}
\caption{\label{Fig_STP_bj}Successful content delivery probability for content $j$  versus the content placement probability. }
\end{figure}
Fig. \ref{Fig_STP_bj} verifies the SCDPs for content $j$ derived in Theorem 1 and Theorem 2 against the content placement probability. The analytical results are obtained from \eqref{muWave_STP}, \eqref{P_m_los_closeform} and \eqref{P_m_nlos_closeform}. The SCDP for {\color{black}an  arbitrary content $j$} is observed to be a monotonically increasing and concave function of the caching placement probability for both $\mu$Wave and $\mathrm{mm}$Wave systems.
%The reason is that the  more \textcolor{black}{SBSs} cache  \textcolor{black}{requested contents}, the more densification gains can be obtained.
Notice that all our derived  analytical results match very well with those ones via Monte Carlo simulations averaged over 2,000 random user drops and  marked by '$\cdot$'.

\begin{figure}
\centering
\includegraphics[width=3.6in, height=3 in]{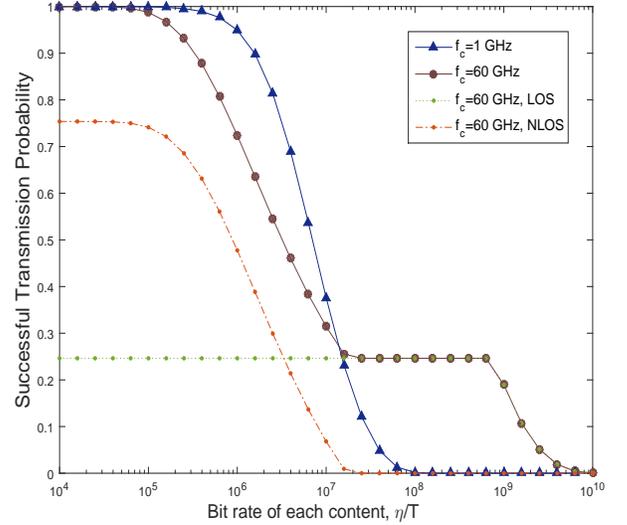}
\caption{\label{Fig_STP_eta}Successful transmission probability, $\lambda_{\mu,\mathrm{mm}}=400 /{\rm{km}}^2$,$b_j=1$.}
\end{figure}

In Fig. \ref{Fig_STP_eta}, we examine the comparison of successful transmission probabilities of $\mu$Wave from  \eqref{muWave_STP} and $\mathrm{mm}$Wave from \eqref{P_m_los_closeform} and \eqref{P_m_nlos_closeform} as \textcolor{black}{bit rate of each content} varies, which corresponds to the case with caching placement probability  $b_j=1$.  It is seen that when content size is small, the $\mu$Wave system shows   better performance than $\mathrm{mm}$Wave, but as the content size increases, the $\mathrm{mm}$Wave system outperforms the $\mu$Wave system for its ability to provide high capacity. The successful $\mathrm{mm}$Wave transmission probability  shows a `ladder drop' effect, and this is because the $\mathrm{mm}$Wave system combines LOS part and NLOS part.   The LOS effect is limited to the region within the distance $D_{\rm L})$ while NLOS has a much wider coverage, so when the required content size is small, the performance is dominated by the NLOS part. However, the NLOS part cannot provide high capacity due to the much larger path loss exponent $\alpha_{\rm N}$, so its performance drops steeply as the \textcolor{black}{bit rate of each content} increases.

Next, in Figs. \ref{Fig_SIR_M_N=1}-\ref{Fig_SIR_S_N=1}, {\color{black}{we compare the performance of the two proposed content placement schemes with the close-form optimal solution \cite{wen2016cache} and the intuitive MPC scheme {\cite{liu2016cache}} in the $\mu$Wave single-antenna case.}} \textcolor{black}{Note that in the general multi-antenna setting, the close-form optimal content placement is still unknown}. The SCDP with   different caching capacity $M$ is shown in Fig. \ref{Fig_SIR_M_N=1}.  \textcolor{black}{It is observed that the CCEO algorithm achieves exactly the same performance as the known optimal solution in \cite{wen2016cache}, and the proposed TS scheme  provides close-to-optimal and significantly better performance than the MPC solution, especially when $\gamma$ is large and the caching capacity $M$ is small}. The MPC solution is the worst caching scheme because it ignores the content diversity which is particularly  important when the content popularity is more uniform.  Fig. \ref{Fig_SIR_S_N=1} shows the SCDP with different content sizes $\eta$. \textcolor{black}{It is found that the SCDP of the TS scheme is closer to the optimum when the $\eta/T$  is large. However, as the bit rate of each content $\eta/T$ increases, both TS and MPC schemes become very close to the optimal solution.}
\begin{figure}
\centering
\includegraphics[width=3.6in, height=3 in]{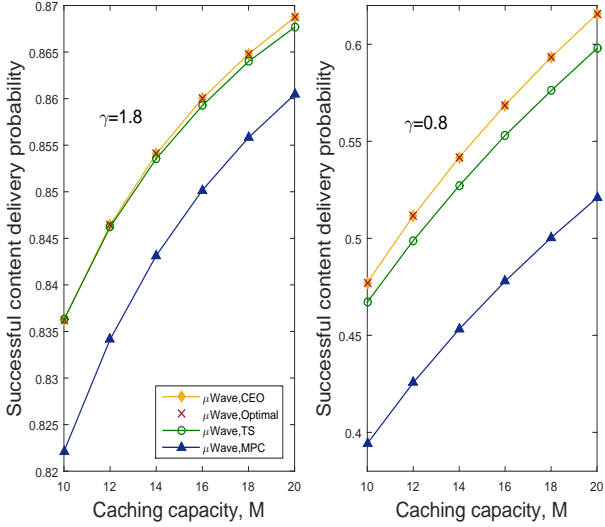}
\caption{\label{Fig_SIR_M_N=1}Successful content delivery probability for $\mu$Wave single antennas.}
\end{figure}

\begin{figure}
\centering
\includegraphics[width=3.6in, height=3 in]{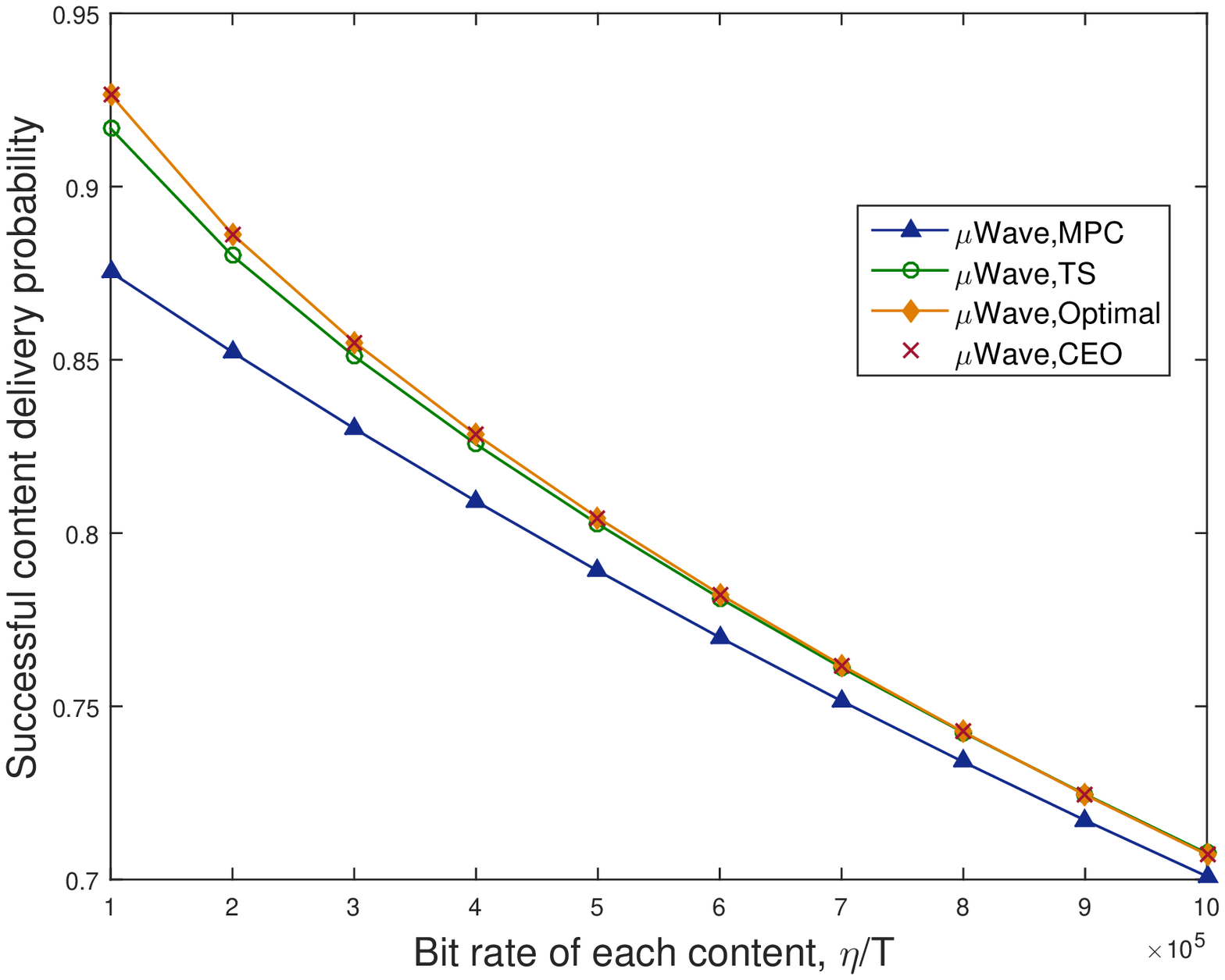}
\caption{\label{Fig_SIR_S_N=1}Successful content delivery probability for $\mu$Wave single antennas,  $\gamma=0.6$.}
\end{figure}

%%%=====multi-antenna-===========

\begin{figure}
\centering
\includegraphics[width=3.6in, height=3 in]{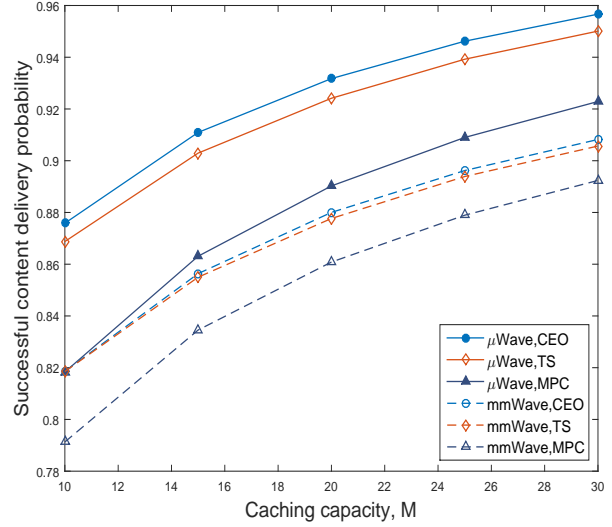}
\caption{\label{Fig_M_N=2}The impact of $M$ on the successful content delivery probability, $\gamma=1.5$.}
\end{figure}

Fig. \ref{Fig_M_N=2} shows the SCDP comparison of various systems with different caching capacities $M$. It shows that both of the proposed content placement schemes perform consistently better than MPC, \textcolor{black}{especially for the 60GHz $\rm{mm}$Wave, the SCDP of the TS scheme is close to that of the CCEO algorithm}. The results also indicate that $\mu$Wave always has a superior performance than the 60GHz $\mathrm{mm}$Wave with the same SBS density of $600 /{\rm{km}}^2$.

\begin{figure}
\centering
\includegraphics[width=3.6in, height=3 in]{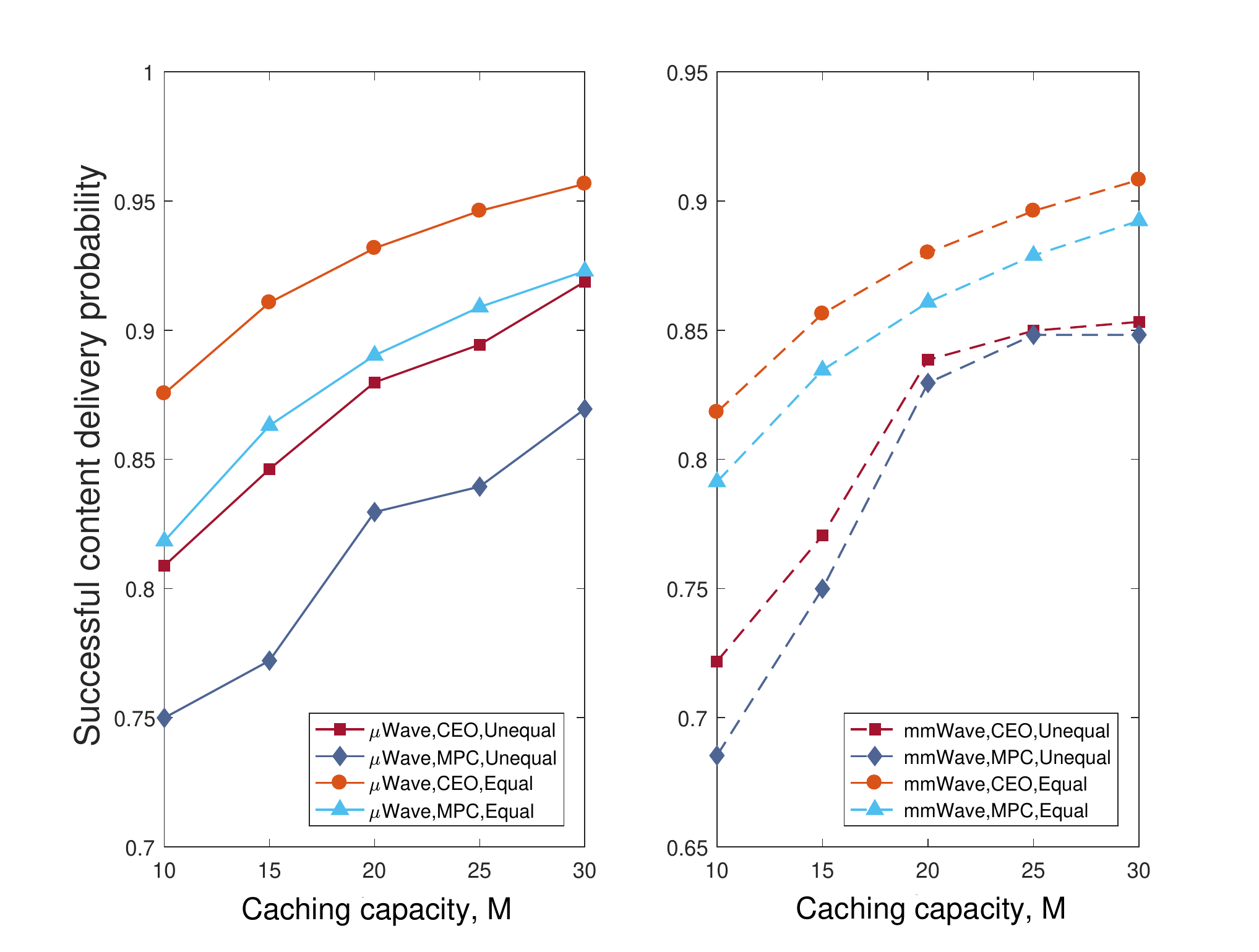}
\caption{\label{unequal}The impact of $M$ on the successful content delivery probability with unequal content size, $\gamma=1.5$.}
\end{figure}

\textcolor{black}{Fig. \ref{unequal} shows the SCDP comparison of various systems versus the  caching capacities $M$ with different content sizes. We generate a random set of of content size $S=\{ s_1,...,s_j,...s_J\}$, where $s_j $ denotes the content size of $f_j$.  For simplicity,  $s_j$ is chosen to be  $1$ or $2$ with equal probability of $0.5$ in our simulation. The caching probability satisfies  $\sum_{j=1}^J b_j \times s_j \leq M$. It is shown that  in the unequal-size content   case,  CEO still greatly outperforms MPC, following a similar trend as the equal-size content case.}

\begin{figure}
\centering
\includegraphics[width=3.6in, height=3 in]{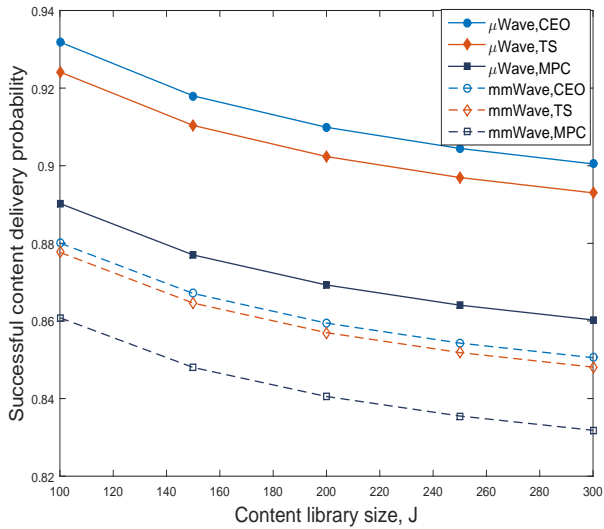}
\caption{\label{Fig_J_N=2}The impact of $J$ on the successful content delivery probability, $\gamma=1.5$, $M=20$.}
\end{figure}

Fig. \ref{Fig_J_N=2} studies the impact of content library size on SCDPs of different systems. It is seen  that as the library size $J$ increases, the SCDP drops rapidly. The gap between the proposed content placement schemes and the MPC scheme remain stabilized when the library size increases.

\begin{figure}
\centering
\includegraphics[width=3.6in, height=3 in]{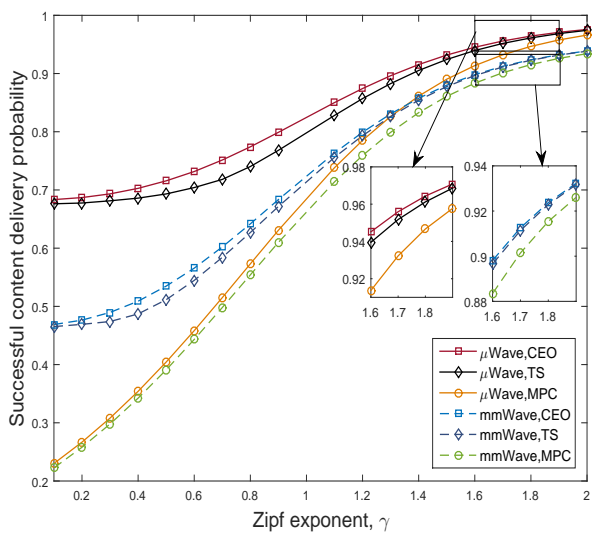}
\caption{\label{Fig_gamma_N=2} The impact of Zipf exponent $\gamma$ on the
successful content delivery probability.}
\end{figure}

Fig. \ref{Fig_gamma_N=2} compares the SCDPs for the two proposed content placement schemes against Zipf exponent $\gamma$. It can be seen that the SCDP increases with $\gamma$ because caching is more effective when the content reuse is high. In the high-$\gamma$ regime of both $\mu$Wave and $\mathrm{mm}$Wave systems, the content request probabilities for the first few most popular content are large, and SCDPs of both proposed placement schemes almost coincide. \textcolor{black}{It is noteworthy that the proposed TS placement scheme achieves performance close to the CCEO algorithm, especially in the $\mu$Wave system and at low and high $\gamma$ regimes}.
%$\mu$Wave under MPC+CD scheme appears clearly better than others.

\begin{figure}
\centering
\includegraphics[width=3.6in, height=3 in]{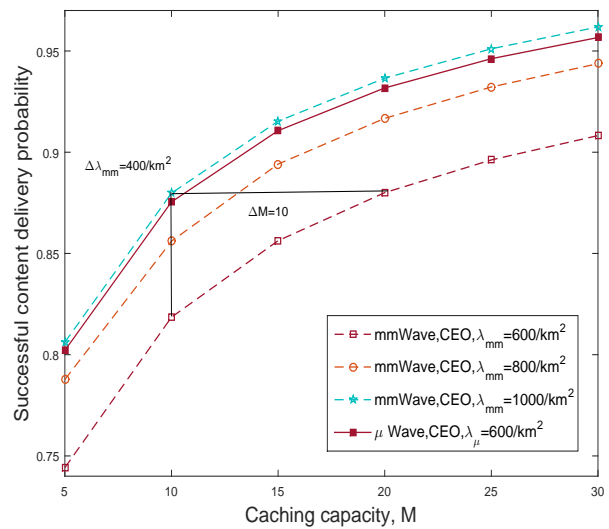}
\caption{\label{Fig_M_compare_N=2}Cache-density tradeoff , $\gamma=1.5$.}
\end{figure}

Finally, we investigate the cache-density tradeoff and its implication on the comparison of  $\mu$Wave and $\mathrm{mm}$Wave systems. \textcolor{black}{The CEO placement scheme is used.} Fig. \ref{Fig_M_compare_N=2}  demonstrates the SCDPs with  different caching capacity $M$, SBS densities $\lambda_{\mu}$ and $\lambda_\mathrm{mm}$. It is also observed that the $\mu$Wave channel is usually better than the $\mathrm{mm}$Wave channel \textcolor{black}{when $\lambda=600 /{\rm{km}}^2$}, so with the same SBS density, $\mu$Wave achieves higher SCDP. To achieve performance comparable to that of the $\mu$Wave system with SBS density of $600 /{\rm{km}}^2$, the $\mathrm{mm}$Wave system needs to deploy SBSs with a much higher density of $1000 /{\rm{km}}^2$, but the extra density of $\triangle \lambda_{\mathrm{mm}}$ =400 /km$^2$ is too costly to afford. Fortunately, by increasing the caching capacity from 10 to \textcolor{black}{20}, the $\mathrm{mm}$Wave system can achieve the same SCDP of \textcolor{black}{91}\% as the  $\mu$Wave system while keeping the same density of $600 /{\rm{km}}^2$. This result shows great promise of cache-enabled small cell systems because it is possible to trade off the relatively cheap storage for reduced  expensive infrastructure.

%\textcolor{black}{We obtain the same observation like before that increasing caching capacity is a low-cost and effective solution to close the gap between {$\mu$Wave} and mmWave systems.}

\section{Conclusion}
In this paper, we have investigated the performance of caching in $\mu$Wave and mmWave multi-antenna dense networks to improve the efficiency of content delivery. Using stochastic geometry, we have analyzed the successful content delivery probabilities and demonstrated the impact of various system parameters.  \textcolor{black}{We designed two novel caching schemes to maximize the successful content delivery probability with moderate to low complexities. The proposed CCEO algorithm can achieve near-optimal performance while the proposed TS scheme demonstrates performance close to CCEO with further reduced complexity}. An important implication of this work is that to reduce the performance gap between the \textcolor{black}{$\mu$Wave} and \rm{mm}Wave systems, increasing caching capacity is a low-cost and effective solution compared to the traditional measures such as using more antennas or increasing SBS density. As a promising future direction, to study cooperative caching in a \textcolor{black}{multi-band} {$\mu$Wave} and mmWave system could further reap the benefits of both systems.

\section*{Appendix A: Proof of Theorem 1}
\label{App:theo_1}
\renewcommand{\theequation}{A.\arabic{equation}}
\setcounter{equation}{0}

Based on \eqref{SCDP_muWave}, $\mathcal{P}_{\mathrm{SCD}}^{\mu}$ is calculated as
\begin{align}\label{A_1}
&\mathcal{P}_{\mathrm{SCD}}^{\mu}=\sum\limits_{j{\rm{ = }}1}^J {a_j} {\rm{Pr}}\left(\frac{{{P_\mu}{h^{\mu}_j}L\left( {\left| {{X^{\mu}_j }} \right|} \right)}}{{\mathcal{I}^\mu_j + \overline{\mathcal{I}}^\mu_j+ {\sigma^2_\mu}}} >\varphi_\mu\right) \nonumber\\
& = \sum\limits_{j{\rm{ = }}1}^J {a_j} \int_0^\infty  \underbrace{\Pr \left( { \frac{{{P_\mu}{h^{\mu}_j}L\left( x  \right)}}{{\mathcal{I}^\mu_j + \overline{\mathcal{I}}^\mu_j+ {\sigma^2_\mu}}} >\varphi_\mu } \right)}_{{P_{\operatorname{cov} }^\mu (x,b_j)}} f_{\left| {{X}}_j^{\mu} \right|}(x) dx ,
\end{align}
where $P_{\operatorname{cov} }^\mu (x,b_j)$ is the conditional coverage probability, and $f_{\left| {{X}}_j^{\mu} \right|}(x)$ is the PDF of the distance $\left| {{X}}_j^{\mu} \right|$. Then, we derive    $P_{\operatorname{cov} }^\mu (x,b_j)$ as
\begin{align}\label{A2}
& P_{\operatorname{cov} }^\mu (x,b_j)= \Pr \left( { \frac{{{P_\mu}{h^{\mu}_j}\beta_\mu x^{-\alpha_\mu}}}{{\mathcal{I}^\mu_j + \overline{\mathcal{I}}^\mu_j+ {\sigma^2_\mu}}} >\varphi_\mu } \right) \nonumber \\
&   = \int_0^\infty  {\Pr (h^{\mu}_j > \frac{{{\varphi_\mu}\left( {\tau {\text{  +  }}\sigma_\mu^2} \right){x^{{\alpha_\mu}}}}}{{{P_\mu} \beta_\mu}}){{d}}\Pr ({\mathcal I_{total} } \leq \tau )}  \nonumber\\
 &\hspace{-0.4 cm} =\int_0^\infty  {{e^{ - \frac{{\left( {\tau {\text{  +  }}\sigma_\mu^2} \right){\varphi_\mu}{x^{{\alpha_\mu}}}}}{{{P_\mu}\beta_\mu}}}}} \sum\limits_{n = 0}^{{N_\mu} - 1} {\frac{{{{\left( {\frac{{\left( {\tau {\text{  +  }}\sigma_\mu^2} \right){\varphi _\mu}{x^{{\alpha_\mu}}}}}{{{P_\mu}\beta_\mu}}} \right)}^n}}}{{n!}}} d\Pr \left( {{\mathcal {I}_{total}}\leq \tau } \right)
\end{align}
where $\mathcal{I}_{total}=\mathcal{I}^\mu_j + \overline{\mathcal{I}}^\mu_j$. Note that
\begin{align}\label{A3}
&{\left. {\frac{{{d^n}\left( {\exp \left( { - \frac{{\left( {\tau+\sigma_\mu^2} \right){\varphi_\mu}{\nu}}}{{{P_\mu}\beta_\mu}}} \right)} \right)}}{{d{\nu^n}}}} \right|_{\nu= {x^{{\alpha_\mu}}}}} \nonumber \\
&  = {\left( { - \frac{{\left( {\tau +\sigma_\mu^2} \right){\varphi_\mu}}}{{{P_\mu}\beta_\mu}}} \right)^n}\exp \left( { - \frac{{\left( {\tau {+}\sigma_\mu^2} \right){\varphi_\mu}{\nu}}}{{{P_\mu}\beta_\mu}}} \right).
\end{align}
By  using \eqref{A3}, \eqref{A2} can be rewritten as
\begin{align}\label{A4}
& P_{\operatorname{cov} }^\mu (x,b_j)  \nonumber\\
&=  \sum\limits_{n = 0}^{{N_\mu} - 1} {\frac{{{x^{n{\alpha_\mu}}}}}{{n!{{( - 1)}^n}}}} {\left. {\frac{{{d^n}\left[ {\exp ( - \frac{{\nu {\varphi_\mu}\sigma_\mu^2}}{{{P_\mu}\beta_\mu}})\mathcal{L}_{\mathcal{I}^\mu_j}(\frac{{{\varphi_\mu}\nu}}{{{P_\mu}\beta_\mu}}) \mathcal{L}_{\overline{\mathcal{I}}^\mu_j} (\frac{{{\varphi _\mu}\nu}}{{{P_\mu}\beta_\mu}})} \right]}}{{d{\nu^n}}}} \right|_{\nu  = {x^{{\alpha_\mu}}}}},
\end{align}
where $\mathcal{L}_{\mathcal{I}^\mu_j}\left(\cdot\right)$ is the Laplace transform of the PDF $\mathcal{I}^\mu_j$, and $\mathcal{L}_{\overline{\mathcal{I}}^\mu_j}\left(\cdot\right)$ is the Laplace transform of the PDF $\overline{\mathcal{I}}^\mu_j$.
Then $\mathcal{L}_{\mathcal{I}^\mu_j}\left(s\right)$ is given by
\begin{align}\label{A5}
& \mathcal{L}_{\mathcal{I}^\mu_j}\left(s\right)= {\mathbb{E}_{\Phi^\mu_j}}\left[ {\exp \left( { - s \sum\nolimits_{i \in {\Phi^\mu_j }\backslash o} {{P_{\mu}}
{h_{i,o}}L\left( {\left|X_{i,o}\right|} \right)}} \right)} \right] \nonumber \\
& = {\mathbb{E}_{\Phi^\mu_j}}\left[ {{\prod _{i \in \Phi^\mu_j \backslash \{ o \} }}{\mathbb{E}_{{h_{i,o}}}}\left\{ {\exp \left( { - s{P_\mu} h_{i,o} L\left( {\left|X_{i,o}\right|} \right)} \right)} \right\}} \right] \nonumber \\
& =   \exp \left[- {\int_x^\infty  {\left(1 - {\mathbb{E}_{{h_{i,o}}}}\left\{ {\exp \left( { - s{P_\mu} h_{i,o} \beta_\mu {r^{ - {\alpha_\mu}}}} \right)} \right\}\right)2\pi {b_j}{\lambda_\mu}rdr} } \right] \nonumber \\
&  = \exp \left[ { - 2\pi {b_j}{\lambda _\mu}\int_x^\infty  { \left(1 - \frac{1}{{1 + s{P_\mu}\beta {r^{ - {\alpha_\mu}}}}} \right)rdr} } \right].
\end{align}
Likewise, $\mathcal{L}_{\overline{\mathcal{I}}^\mu_j}\left(s\right)$ is given by
\begin{align}\label{A6}
& \mathcal{L}_{\overline{\mathcal{I}}^\mu_j}\left(s\right)= {\mathbb{E}_{\overline{\mathcal{I}}^\mu_j}}\left[ {\exp \left( { - s\sum\nolimits_{k \in {\overline{\Phi}^\mu_j}} {{P_{\mu}}
{h_{k,o}}L\left( \left|X_{k,o}\right| \right)}} \right)} \right] \nonumber \\
&= {\mathbb{E}_{\overline{\mathcal{I}}^\mu_j}}\left[ {{\prod _{k \in \overline{\mathcal{I}}^\mu_j }}{\mathbb{E}_{{h_{k,o}}}}\left\{ {\exp \left( { - s{P_\mu} h_{k,o} L\left( {\left|X_{k,o}\right|} \right)} \right)} \right\}} \right] \nonumber \\
& =   \exp \Bigg[- \int_0^\infty  \left(1 - {\mathbb{E}_{{h_{k,o}}}}\left\{ {\exp \left( { - s{P_\mu} h_{k,o} \beta_\mu {r^{ - {\alpha_\mu}}}} \right)} \right\}\right) \nonumber\\
&\qquad\qquad  \times 2\pi \left( {1 - {b_j}} \right){\lambda _\mu}rdr  \Bigg] \nonumber \\
& = \exp \left[ { - 2\pi \left( {1 - {b_j}} \right){\lambda_\mu}\int_0^\infty  {\left(1 - \frac{1}{{1 + s{P_\mu}\beta_\mu {r^{ - {\alpha_\mu}}}}}\right) rdr} } \right].
\end{align}

Substituting \eqref{A5} and \eqref{A6} into \eqref{A4}, after some manipulations, we can obtain the desired result \eqref{cov_pro_muWave}.

\section*{Appendix B: Proof of Theorem 2}
\label{App:theo_2}
\renewcommand{\theequation}{B.\arabic{equation}}
\setcounter{equation}{0}
Based on \eqref{SINR_mmWave} and \eqref{SCDP_mmWave}, the SCDP for a LOS $\mathrm{mm}$Wave link can be derived as
\begin{align}\label{P_m_los}
&\hspace{-0.3 cm} \mathcal{P}_{\mathrm{SCD}}^{\rm{mm,L}}= \int_0^{D_L} \Pr \left(  \frac{{{P_\mathrm{mm}} G_{\rm{mm}} \beta_{\rm{mm}} {y^{ - {\alpha _{\text{L}}}}} }}{{ {\sigma^2_\mathrm{mm}}}} > {\varphi_\mathrm{mm}} \right) f_{\left| {{Y^{\mathrm{mm}}_j }}\right|}(y)dy \nonumber \\
&   =  \mathbbm{1} (D_L<{d_{\text{L}}})\int_0^{D_L}  f_{\left| {{Y^{\mathrm{mm}}_j }}\right|}(y) dy  \nonumber\\
&\qquad\qquad\qquad \quad+\mathbbm{1}(D_L > {d_{\text{L}}})\int_0^{{d_{\text{L}}}} f_{\left| {{Y^{\mathrm{mm}}_j }}\right|}(y)  dy \nonumber \\
&   = 1- e^{ - \left( \min{\left( D_L,{d_{\rm{L}}} \right)} \right)^2 \pi b_j \lambda_{\rm{mm}} },
\end{align}
where $ d_\text{L}= {\left( {\frac{{P_\mathrm{mm} G_\mathrm{mm}\beta_{\rm{mm}} }}{{{\varphi _\mathrm{mm}}\sigma _\mathrm{mm}^2}}} \right)}^{ \frac{1}{\alpha _{\text{L}}}} $, $f_{\left| {{Y^{\mathrm{mm}}_j }}\right|}(y)$ is the PDF of the distance $\left| {{Y^{\mathrm{mm}}_j }}\right|$ between a typical user and its serving $\mathrm{mm}$Wave SBS , which is given by {\color{black}\cite{jo2012heterogeneous}
\begin{align}\label{probability_m}
f_{\left| {{Y^{\mathrm{mm}}_j }}\right|}(y)= 2\pi b_j \lambda_{\rm{mm}} y{e^{ - \pi b_j \lambda_{\rm{mm}} {y^2}}},~~~ y\geq 0.
\end{align}}

Similarly, the SCDP for a NLOS $\mathrm{mm}$Wave link can be derived as
\begin{align}\label{P_m_nlos}
&\mathcal{P}_{\mathrm{SCD}}^{\rm{mm,N}}= \int_{D_L}^\infty  {\Pr \left( {\frac{{{P_{\rm{mm}}}{G_{\rm{mm}}}\beta_{\rm{mm}} {y^{ - {\alpha _{\text{N}}}}}}}{{\sigma _{\rm{mm}}^2}} > {\varphi _{\rm{mm}}}} \right)} f_{\left| {{Y^{\mathrm{mm}}_j }}\right|}(y)  dy \nonumber \\
&   =  \mathbbm{1} (D_L<{d_{\text{N}}})\int_{D_L}^{d_{\text{N}}}  f_{\left| {{Y^{\mathrm{mm}}_j }}\right|}(y) dy  \nonumber\\
&=e^{ - {D_L^2}\pi b_j \lambda_{\rm{mm}} }- e^{ - \left(\max{\left( {D_L,{d_{\rm{N}}}}\right)}\right)^2 \pi b_j {\lambda_{\rm{mm}}}},
\end{align}
where $ d_\text{N}= {\left( {\frac{{P_\mathrm{mm} G_\mathrm{mm}\beta_{mm} }}{{{\varphi_\mathrm{mm}}\sigma _\mathrm{mm}^2}}} \right)}^{ \frac{1}{\alpha _{\text{N}}}} $. Thus, we obtain the SCDP expressions for a LOS/NLOS $\mathrm{mm}$Wave link.

\section*{Appendix C: Proof of Theorem 3}
\label{App:theo_2}
\renewcommand{\theequation}{C.\arabic{equation}}
\setcounter{equation}{0}

Let $f_1\left(\varepsilon\right)$ denote the objective function of the problem \eqref{Problem_formulation_muWave_sub}. We can obtain the first-order and the second-order derivatives of $f_1\left(\varepsilon\right)$ with respective to (w.r.t.) $\varepsilon$ as
\begin{align}\label{first_order_varep}
\frac{{\partial {f_1}\left( \varepsilon  \right)}}{{\partial \varepsilon }}& = \frac{1}{{{J^{1 - \gamma }} - 1}}\left[ {\left( {\ell _o^\mu  - 1} \right)\left( {1 - \gamma } \right){\varepsilon ^{ - \gamma }}} \right.\nonumber\\
& \left. { + \left( {1 - \gamma } \right)\left( {1 - \frac{1}{\varpi }} \right){{\left( {\varepsilon  + \frac{{(1 - \varepsilon )}}{\varpi }} \right)}^{ - \gamma }}} \right] ,
\end{align}
and
\begin{align}\label{first_order_varep}
\frac{{\partial^2{f_1}\left( \varepsilon \right)}}{{\partial \varepsilon^2}}& = \frac{1}{{{J^{1 - \gamma }} - 1}}\left( {1 - \gamma } \right)\left( { - \gamma } \right) \nonumber\\
&  \left[ {\left( {\ell _o^\mu  - 1} \right){\varepsilon ^{ - \gamma  - 1}} + {{\left( {1 - \frac{1}{\varpi }} \right)}^2}{{\left( {\varepsilon  + \frac{{(1 - \varepsilon )}}{\varpi }} \right)}^{ - \gamma  - 1}}} \right],
\end{align}
respectively. Note that $\ell_o^{\mu} \geq 1$  and $\frac{{1 - \gamma }}{{{J^{1 - \gamma }} - 1}}>0$, so we get $\frac{{\partial^2{f_1}\left( \varepsilon \right)}}{{\partial \varepsilon^2}} \leq 0$, which means that ${f_1}\left( \varepsilon \right)$ is a concave function w.r.t. $\varepsilon$. By setting $\frac{{\partial {f_1}\left( \varepsilon  \right)}}{{\partial \varepsilon }}$ to zero, we obtain the stationary point as
\begin{align}\label{sta_point}
\varepsilon_o={\left( {\left( {{{\left( {\frac{{{\ell _o^{\mu}} - 1}}{{{\varpi ^{ - 1}} - 1}}} \right)}^{ - 1/\gamma }} - 1} \right)\varpi  + 1} \right)^{ - 1}}.
\end{align}
Note that $0\leq \varpi \leq 1$, and ${\frac{{{\ell _o^{\mu}} - 1}}{{{\varpi ^{ - 1}} - 1}}}\geq0$, we have $\varepsilon_o \geq 0$. To obtain the optimal $\varepsilon^*$, we need to consider the following cases:
\begin{itemize}
%  \item Case 1: $\varepsilon_o \leq \theta$. In this case, $\frac{{\partial {f_1}\left( \varepsilon  \right)}}{{\partial \varepsilon }} \leq 0$ for $\varepsilon\in\left[\theta,1\right]$, thus the optimal solution of the problem $\bf{P2-1}$  is $\varepsilon^*=\theta$.
  \item Case 1: $0\leq \varepsilon_o < 1$. In this case,  the optimal solution of the problem \eqref{Problem_formulation_muWave_sub} is $\varepsilon^*=\varepsilon_o$.
  \item Case 2: $ \varepsilon_o \geq 1$. In this case,  $\frac{{\partial {f_1}\left( \varepsilon  \right)}}{{\partial \varepsilon }} \geq 0$ for $\varepsilon\in\left[0,1\right]$, and thus the optimal solution of the problem \eqref{Problem_formulation_muWave_sub}  is $\varepsilon^*=1$.
\end{itemize}
Based on the above cases, we obtain \eqref{optimal_P2_1} and complete the proof.

\section*{Appendix D: Newton$^\prime$s Method to optimize $\varpi$ in \eqref{Problem_formulation_muWave_sub2}}
\label{App:theo_3}
\renewcommand{\theequation}{D.\arabic{equation}}
\setcounter{equation}{0}
We propose Newton's Method to solve the non-convex problem \eqref{Problem_formulation_muWave_sub2} with fast convergence. Based on \eqref{obj_approx_muWave}, the first-order derivative of ${\widetilde{\mathcal{P}}}_{\mathrm{SCD}}^{\mu}$ is given by
\begin{align}\label{P_scdp_1_order}
&\frac{\partial {\widetilde{\mathcal{P}}}_{\mathrm{SCD}}^{\mu}}{{\partial \varpi }} = \mathcal{P}_{j,{\text{SCD}}}^\mu \left( 1 \right)\frac{{{M^{1 - \gamma }}}}{{{J^{1 - \gamma }} - 1}}\left( {1 - \gamma } \right)\varepsilon _o^{ - \gamma }\frac{{\partial {\varepsilon _o}}}{{\partial \varpi }} \nonumber\\
&+ \frac{{\partial \mathcal{P}_{j,{\text{SCD}}}^\mu \left( \varpi  \right)}}{{\partial \varpi }}\frac{{{M^{1 - \gamma }}}}{{{J^{1 - \gamma }} - 1}}[{\left( {{\varepsilon _o} + \frac{{(1 - {\varepsilon _o})}}{\varpi }} \right)^{1 - \gamma }} - {\varepsilon_o ^{1 - \gamma }}] \nonumber\\
&+ \mathcal{P}_{j,{\text{SCD}}}^\mu \left( \varpi  \right)\frac{{{M^{1 - \gamma }}}}{{{J^{1 - \gamma }} - 1}}(1 - \gamma )\left\{ { - \varepsilon _o^{ - \gamma }\frac{{\partial {\varepsilon _o}}}{{\partial \varpi }}} \right. \nonumber\\
& +\left. {{{\left( {{\varepsilon _o} + \frac{{(1 - {\varepsilon _o})}}{\varpi }} \right)}^{ - \gamma }}\left[ {\frac{{\partial {\varepsilon _o}}}{{\partial \varpi }}(1 - {\varpi ^{ - 1}}){\text{ + }}\frac{{{\varepsilon _o} - 1}}{{{\varpi ^2}}}} \right]} \right\} ,
\end{align}
where
\begin{align}\label{var_first_order}
\frac{\partial \varepsilon_o}{\partial \varpi}&=-(\varepsilon_o)^2\Big({{{\big( {\frac{{{\ell _o} - 1}}
{{{\varpi ^{ - 1}} - 1}}} \big)}^{ - 1/\gamma }} - 1}  \nonumber\\
&-(\gamma\varpi)^{-1}\left(\ell_o-1\right)^{-\frac{1}{\gamma}}\Big({\varpi^{-1}-1} \Big)^{\frac{1}{\gamma}-1}  \Big),
\end{align}
and
\begin{align}\label{SCDP_Pr_first_order}
\frac{\partial \mathcal{P}_{j,\mathrm{SCD}}^{\mu}\left(\varpi\right)}{\partial \varpi}&=\int_0^\infty  {P_{\operatorname{cov} } (x )}\Big(2\pi {\lambda_\mu }x e^{ { - \pi \varpi{\lambda_\mu } {x^2}} } \nonumber\\
&\quad\qquad-2\pi^2\varpi {\lambda_\mu^2}x^3 e^{ { - \pi \varpi{\lambda_\mu } {x^2}} }\Big) dx,
\end{align}
respectively. In \eqref{SCDP_Pr_first_order}, to simplify the computation, we let ${P_{\operatorname{cov} } (x )}\approx {P_{\operatorname{cov} }^\mu (x,0 )}$, based on the fact that the interference $\mathcal{I}^\mu_j + \overline{\mathcal{I}}^\mu_j $ can be approximated as $\sum\nolimits_{k \in {\Phi_\mu}} {{P_{\mu}} {h_{k,o}}L\left( \left|X_{k,o}\right| \right)}$, particularly in the dense small cell scenarios~\cite{XinqinLin_2014}. Similarly, the second-order derivative of ${\widetilde{\mathcal{P}}}_{\mathrm{SCD}}^{\mu}(\eta,T)$ is given by
\begin{align}\label{P_scdp_2_order}
&\frac{{{\partial^{\rm{2}}{\widetilde{\mathcal{P}}}_{\mathrm{SCD}}^{\mu}}}}{{\partial {\varpi ^{\rm{2}}}}}=\mathcal{P}_{j,{\text{SCD}}}^\mu \left( 1 \right)\frac{{{M^{1 - \gamma }}}}{{{J^{1 - \gamma }} - 1}}\left( {1 - \gamma } \right)\varepsilon _o^{ - \gamma }( - \gamma \varepsilon _o^{ - 1}{(\frac{{\partial {\varepsilon _o}}}{{\partial \varpi }})^2} + \frac{{{\partial ^2}{\varepsilon _o}}}{{\partial {\varpi ^2}}})\nonumber\\
&+ \frac{{{\partial ^2}\mathcal{P}_{j,{\text{SCD}}}^\mu \left( \varpi  \right)}}{{\partial {\varpi ^2}}}\frac{{{M^{1 - \gamma }}}}{{{J^{1 - \gamma }} - 1}}[{\left( {{\varepsilon _o} + \frac{{(1 - {\varepsilon _o})}}{\varpi }} \right)^{1 - \gamma }} - {\varepsilon ^{1 - \gamma }}]\nonumber\\
&   + 2\frac{{\partial \mathcal{P}_{j,{\text{SCD}}}^\mu \left( \varpi  \right)}}{{\partial \varpi }}\frac{{{M^{1 - \gamma }}}}{{{J^{1 - \gamma }} - 1}}(1 - \gamma )\left\{ { - \varepsilon _o^{ - \gamma }\frac{{\partial {\varepsilon _o}}}{{\partial \varpi }}} \right. \nonumber\\
&+\left. {{{\left( {{\varepsilon _o} + \frac{{(1 - {\varepsilon _o})}}{\varpi }} \right)}^{ - \gamma }}\left[ {\frac{{\partial {\varepsilon _o}}}{{\partial \varpi }}(1 - {\varpi ^{ - 1}}){\text{ + }}\frac{{{\varepsilon _o} - 1}}{{{\varpi ^2}}}} \right]} \right\} \nonumber\\
&   + \mathcal{P}_{j,{\text{SCD}}}^\mu \left( \varpi  \right)\frac{{{M^{1 - \gamma }}}}{{{J^{1 - \gamma }} - 1}}(1 - \gamma )\left\{ {\gamma \varepsilon _o^{ - \gamma  - 1}\frac{{\partial {\varepsilon _o}}}{{\partial \varpi }} - \varepsilon _o^{ - \gamma }\frac{{{\partial ^2}{\varepsilon _o}}}{{\partial {\varpi ^2}}}} \right. \nonumber\\
& - \gamma {\left( {{\varepsilon _o} + \frac{{(1 - {\varepsilon _o})}}{\varpi }} \right)^{ - \gamma  - 1}}{\left[ {\frac{{\partial {\varepsilon _o}}}{{\partial \varpi }}(1 - {\varpi ^{ - 1}}){\text{ + }}\frac{{{\varepsilon _o} - 1}}{{{\varpi ^2}}}} \right]^2}  \nonumber\\
& +{\left( {{\varepsilon _o} + \frac{{(1 - {\varepsilon _o})}}{\varpi }} \right)^{ - \gamma }} \nonumber\\
&  \times \left. {\left[ {\frac{{{\partial ^2}{\varepsilon _o}}}{{\partial {\varpi ^2}}}(1 - {\varpi ^{ - 1}}){\text{ + }}\frac{{\partial {\varepsilon _o}}}{{\partial \varpi }}{\varpi ^{ - 2}}{\text{ + }}\frac{{\frac{{\partial {\varepsilon _o}}}{{\partial \varpi }}{\varpi ^2} - 2\varpi ({\varepsilon _o} - 1)}}{{{\varpi ^4}}}} \right]} \right\},
\end{align}
where
\begin{align}\label{varep_second_order}
\frac{\partial ^2\varepsilon_o}{\partial \varpi^2}&=2\varepsilon_o^{-1}\left(\frac{\partial \varepsilon_o}{\partial \varpi}\right)^2-\varepsilon_o^{2}\left(\ell_o-1\right)^{-\frac{1}{\gamma}}\Bigg[-\gamma^{-1}\varpi^{-2} \nonumber\\
&\Big({\varpi^{-1}-1} \Big)^{\frac{1}{\gamma}-1}+\gamma^{-1}\varpi^{-2}\Big({\varpi^{-1}-1} \Big)^{\frac{1}{\gamma}-1} \nonumber\\
&+\gamma^{-1} (\gamma\varpi)^{-3}\left(\frac{1}{\gamma}-1\right)\Big({\varpi^{-1}-1} \Big)^{\frac{1}{\gamma}-2}\Bigg]
\end{align}
and
\begin{align}\label{p_cov_second_order}
\frac{\partial^2\mathcal{P}_{j,\mathrm{SCD}}^{\mu}\left(\varpi \right)}{\partial \varpi^2}&=\int_0^\infty  {P_{\operatorname{cov} } (x )}\Big(-4\pi^2 {\lambda_\mu^2}x^3 e^{ { - \pi \varpi{\lambda_\mu } {x^2}} } \nonumber\\
&\quad\qquad+2\pi^3\varpi {\lambda_\mu^3}x^5 e^{ { - \pi \varpi{\lambda_\mu } {x^2}} }\Big) dx.
\end{align}
According to~\cite{yu2013multicell},  the search direction in Newton method can be defined as
\begin{align}\label{Newton_step_eq}
\Delta\varpi=\frac{\partial {\widetilde{\mathcal{P}}}_{\mathrm{SCD}}^{\mu}}{{\partial \varpi }}\Big{/}\left|\frac{{{\partial^{\rm{2}}{\widetilde{\mathcal{P}}}_{\mathrm{SCD}}^{\mu}}}}{{\partial {\varpi ^{\rm{2}}}}}\right|.
\end{align}
Then, $\varpi$ is iteratively updated according to
\begin{align}\label{var_update_eq}
\varpi\left(\varrho+1\right)=\left[\varpi\left(\varrho\right)+\delta_2\left(\varrho\right)\Delta\varpi\right]_0^{\frac{1}{\ell_o^\mu}},
\end{align}
where  $\varrho$ denotes the iteration index, $\ell_o^{\mu}=\frac{\mathcal{P}_{j,\mathrm{SCD}}^{\mu}\left(1 \right)}{\mathcal{P}_{j,\mathrm{SCD}}^{\mu}\left(\varpi \right)}$ is already defined in \eqref{Problem_formulation_muWave_sub}, % with $\varpi=\left[\varpi\left(\varrho\right)+\delta_2\left(\varrho\right)\Delta\varpi\right]^+$,
$\delta_2\left(\varrho\right)$ is the step size that can be determined by backtracking line search~\cite{Convex_Book}. Thus, the optimal $\varpi^*$ can be obtained when reaching convergence.

\section*{Appendix E: Derivation of the search direction in the Newton Method to optimize  $\varpi$ in \eqref{Problem_formulation_muWave_sub2}}
\label{App:theo_3}
\renewcommand{\theequation}{E.\arabic{equation}}
\setcounter{equation}{0}
 We use Newton Method to solve the problem \eqref{Problem_formulation_muWave_sub2}. Here we only derive the search direction that involves the first and second-order derivative, and the rest is similar to Appendix D.

  We first derive the first-order derivative. Similar to \eqref{P_scdp_1_order}, we change $\frac{{\partial \mathcal{P}_{j,{\mathrm{SCD}}}^{\mu}\left( {\varpi  } \right)}}{{\partial \varpi }}$ to $\frac{{\partial \mathcal{P}_{j,{\mathrm{SCD}}}^{{\rm{mm}}}\left( {\varpi  } \right)}}{{\partial \varpi }}$ and get the result below:
\begin{align}\label{SCDP_Pr_first_order_mm}
& \frac{{\partial \mathcal{P}_{j,{\mathrm{ST}}}^{{\rm{mm}}}\left( {\varpi  } \right)}}{{\partial \varpi }} =  - D_\text{L}^2\pi {\lambda _{{\rm{mm}}}}{e^{ - D_\text{L}^2\pi \varpi {\lambda _{{\rm{mm}}}}}} \nonumber\\
&~~~~ + {\left( {\min \left( {{D_\text{L}},{d_{\text{L}}}} \right)} \right)^2}\pi {\lambda _{{\rm{mm}}}}{e^{ - {{\left( {\min \left( {{D_\text{L}},{d_{\text{L}}}} \right)} \right)}^2}\pi \varpi {\lambda _{{\rm{mm}}}}}}  \nonumber\\
&~~~~+ {\left( {\max \left( {{D_\text{L}},{d_{\text{N}}}} \right)} \right)^2}\pi {\lambda _{{\rm{mm}}}}{e^{ - {{\left( {\max \left( {{D_\text{L}},{d_{\rm{N}}}} \right)} \right)}^2}\pi \varpi {\lambda _{{\rm{mm}}}}}}.
\end{align}

Next we focus on   the second-order derivative.   Changing $ \frac{\partial^2\mathcal{P}_{{j,\mathrm{SCD}}}^{\mu}\left(\varpi \right)}{\partial \varpi^2}$ to $\frac{\partial^2\mathcal{P}_{{j,\mathrm{SCD}}}^{\rm{mm}}\left(\varpi \right)}{\partial \varpi^2}$  in \eqref{P_scdp_2_order}  leads to the following result:
\begin{align}\label{p_cov_second_order}
&\frac{\partial^2\mathcal{P}_{{j,\text{SCD}}}^{\rm{mm}}\left(\varpi \right)}{\partial \varpi^2}= D_\text{L}^4{\pi ^2}\lambda _{{\text{mm}}}^2{e^{ - D_\text{L}^2\pi \varpi {\lambda _{{\text{mm}}}}}} \nonumber\\
& ~~- {\left( {\min \left( {{D_\text{L}},{d_{\text{L}}}} \right)} \right)^4}{\pi ^2}\lambda _{{\text{mm}}}^2{e^{ - {{\left( {\min \left( {{D_\text{L}},{d_{\text{L}}}} \right)} \right)}^2}\pi \varpi {\lambda _{{\text{mm}}}}}} \nonumber\\
&~~ - {\left( {\max \left( {{D_\text{L}},{d_{\text{N}}}} \right)} \right)^4}{\pi ^2}\lambda _{{\text{mm}}}^2{e^{ - {{\left( {\max \left( {{D_\text{L}},{d_{\text{N}}}} \right)} \right)}^2}\pi \varpi {\lambda _{{\text{mm}}}}}}.
\end{align}
Therefore,  the search direction in Newton Method can be expressed as
\begin{align}\label{Newton_step_eq}
\Delta\varpi^{\rm{mm}}=\frac{\partial {\widetilde{\mathcal{P}}}_{\mathrm{SCD}}^{\rm{mm}}}{{\partial \varpi }}\Big{/}\left|\frac{{{\partial^{\rm{2}}{\widetilde{\mathcal{P}}}_{\mathrm{SCD}}^{\rm{mm}}}}}{{\partial {\varpi ^{\rm{2}}}}}\right|.
\end{align}
%Then, $\varpi^{\rm{mm}}$ is iteratively updated according to
%\begin{align}\label{var_update_eq}
%\varpi^{\rm{mm}}\left(\varrho+1\right)=\left[\varpi^{\rm{mm}}\left(\varrho\right)+\delta_2\left(\varrho\right)\Delta\varpi\right]_0^{\ell_o^{\rm{mm}}},
%\end{align}
%where  $\varrho$ denotes the iteration index, $\ell_o^{\rm{mm}}=\frac{\mathcal{P}_{j,\mathrm{SCD}}^{\mu}\left(1 \right)}{\mathcal{P}_{j,\mathrm{SCD}}^{\mu}\left(\varpi \right)}$ with $\varpi^{\rm{mm}}=\left[\varpi^{\rm{mm}}\left(\varrho\right)+\delta_2\left(\varrho\right)\Delta\varpi ^{\rm{mm}}\right]^+$, $\delta_2\left(\varrho\right)$ is the step size that can be determined by backtracking line search~\cite{Convex_Book}. Thus, the optimal $\varpi^*$ for  $\mathrm{mm}$Wave can be obtained when reaching convergence.

\bibliographystyle{IEEEtran}

\end{document}